\documentclass[10pt,journal,compsoc]{IEEEtran}
\pdfoutput=1
\usepackage{booktabs}
\usepackage{balance}
\usepackage[mathscr]{euscript}
\usepackage{verbatim}
\usepackage{amssymb,amsmath,color,graphicx, amsbsy, bm}
\usepackage{subcaption}
\usepackage{epsfig}
\usepackage{amsthm}

\usepackage{fixltx2e}
\usepackage{xcolor}

\allowdisplaybreaks[1]

\usepackage{bbm}
\usepackage{amssymb}
\usepackage{algorithmic}       
\usepackage{amsmath}
\usepackage{graphicx}
\usepackage{booktabs}
\usepackage[flushleft]{threeparttable}
\usepackage{multirow}
\usepackage{amsthm}
\usepackage{graphicx,booktabs,multirow}
\usepackage{stmaryrd}
\usepackage{multirow,url}
\usepackage{algorithm}
\usepackage{cite}
\usepackage{color}

\newtheorem{condition}{Condition}
\newtheorem{lemma}{Lemma}

\newtheorem{theorem}{Theorem}
\newtheorem{corollary}{Corollary}
\newtheorem*{assumption*}{Assumption}
\ifodd 1

\newcommand{\com}[1]{\textbf{\color{red} (Comment: #1) }}
\newcommand{\comg}[1]{\textbf{\color{blue} (COMMENT: #1)}}
\newcommand{\response}[1]{\textbf{\color{blue} (RESPONSE: #1)}}
\else

\newcommand{\com}[1]{}
\newcommand{\comg}[1]{}
\newcommand{\response}[1]{}
\fi

\begin{document}

\title{Online Bitrate Selection for Viewport Adaptive 360-Degree Video Streaming}

\author{Ming Tang
	and~Vincent~W.S.~Wong
	\thanks{Ming Tang and Vincent W.S. Wong are with the Department
		of Electrical and Computer Engineering, The University of British Columbia, Vancouver, Canada.\protect\\
		E-mail: \{mingt,~vincentw\}@ece.ubc.ca}}

\IEEEtitleabstractindextext{%
	\begin{abstract}
	360-degree video streaming provides users with immersive experience by letting users determine their field-of-views (FoVs) in real time. To enhance the users' quality of experience (QoE) given their limited bandwidth, recent works have proposed a viewport adaptive 360-degree video streaming model by exploiting the bitrate adaptation in spatial and temporal domains. Under this video streaming model, in this paper, we consider a scenario with a newly generated 360-degree video without viewing history from other users. To maximize the user's QoE, we propose an online bitrate selection algorithm, called OBS360. The proposed online algorithm can adapt to the unknown and heterogeneous users' FoVs and downloading capacities. We prove that the proposed algorithm achieves sublinear dynamic regret under a convex decision set. This suggests that as the number of video segments increases, the performance of the online algorithm approaches the performance of the offline algorithm, where the users' FoVs and downloading capacities are known. We perform simulations with real-world dataset to evaluate the performance of the proposed algorithm. Results show that compared with several existing methods, our proposed algorithm can enhance the users' QoE significantly by improving the viewing bitrate and reducing the inter-segment and intra-segment degradation losses of the users.
\end{abstract}

\begin{IEEEkeywords} 
Adaptive 360-degree video streaming, virtual reality, quality of experience (QoE), online convex optimization, online gradient descent.
\end{IEEEkeywords}}
          
\maketitle

\IEEEdisplaynontitleabstractindextext

\IEEEpeerreviewmaketitle

\IEEEraisesectionheading{\section{Introduction}}

\subsection{Background and Motivation}
With the development of virtual reality (VR) technologies, 360-degree video streaming is becoming increasingly popular. With such 360-degree videos, users can determine their field of views (FoVs) in real time by controlling the direction of the streaming devices, with which users can have immersive video streaming experience. Currently, there are various 360-degree video content providers (e.g., YouTube) and VR devices supporting 360-degree videos (e.g.,  Sony PlayStation VR\cite{Sony}). An example of a 360-degree video streaming in a wireless network is shown in Fig. \ref{fig:model0}.  The  users watching the same 360-degree video may have heterogeneous FoVs (represented by the solid rectangles).

The 360-degree real-time interaction, however, is at the expense of additional bandwidth consumption. This is because all the 360-degree scenes, including the scenes that are being viewed or not being viewed by the users, have to be downloaded in real time in response to the users' interactions. This additional bandwidth consumption imposes the requirement of efficient allocation of radio resources in wireless networks, so as to provide good quality of experience (QoE) video streaming services.  



\begin{figure}[t]
	\centering
	\includegraphics[height=2.6cm]{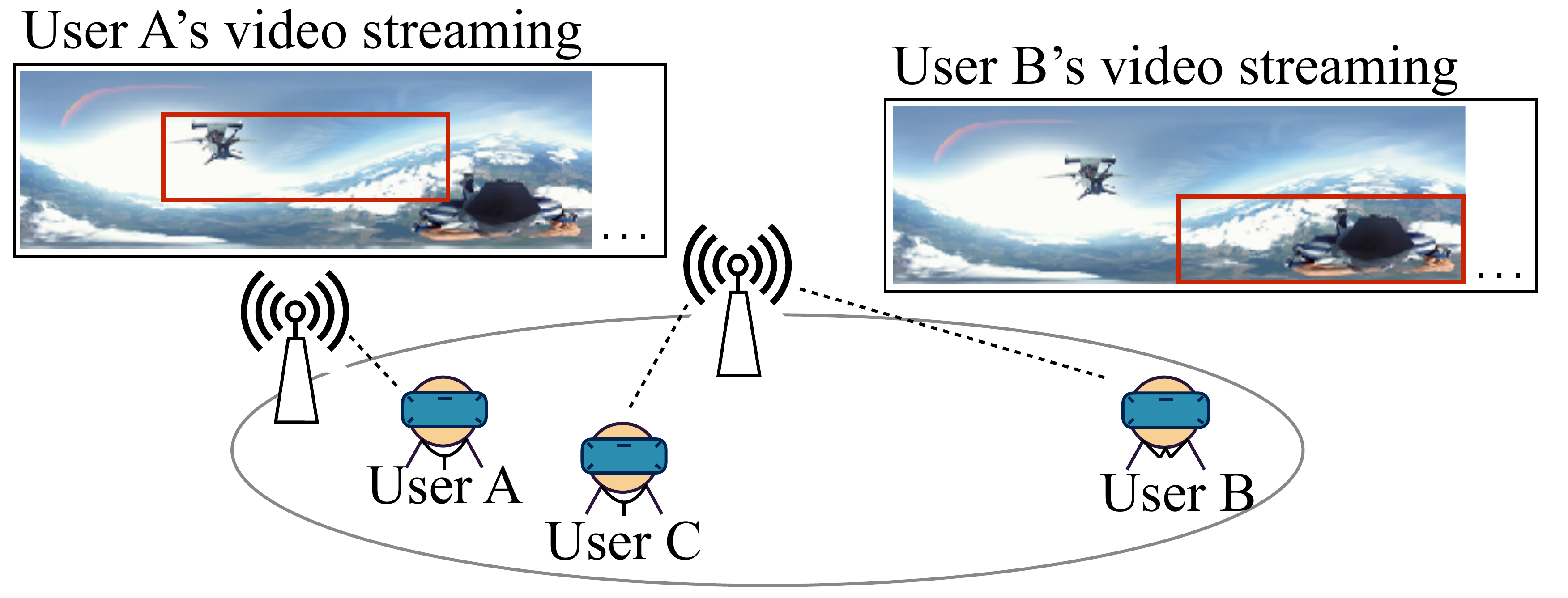}
	\caption{{360-degree video streaming in a wireless network. Users A and  B watch the same video, but they have heterogeneous FoVs. 
		}}\label{fig:model0}
	\end{figure}
	
	\begin{figure}
		\centering
		\includegraphics[height=2.7cm]{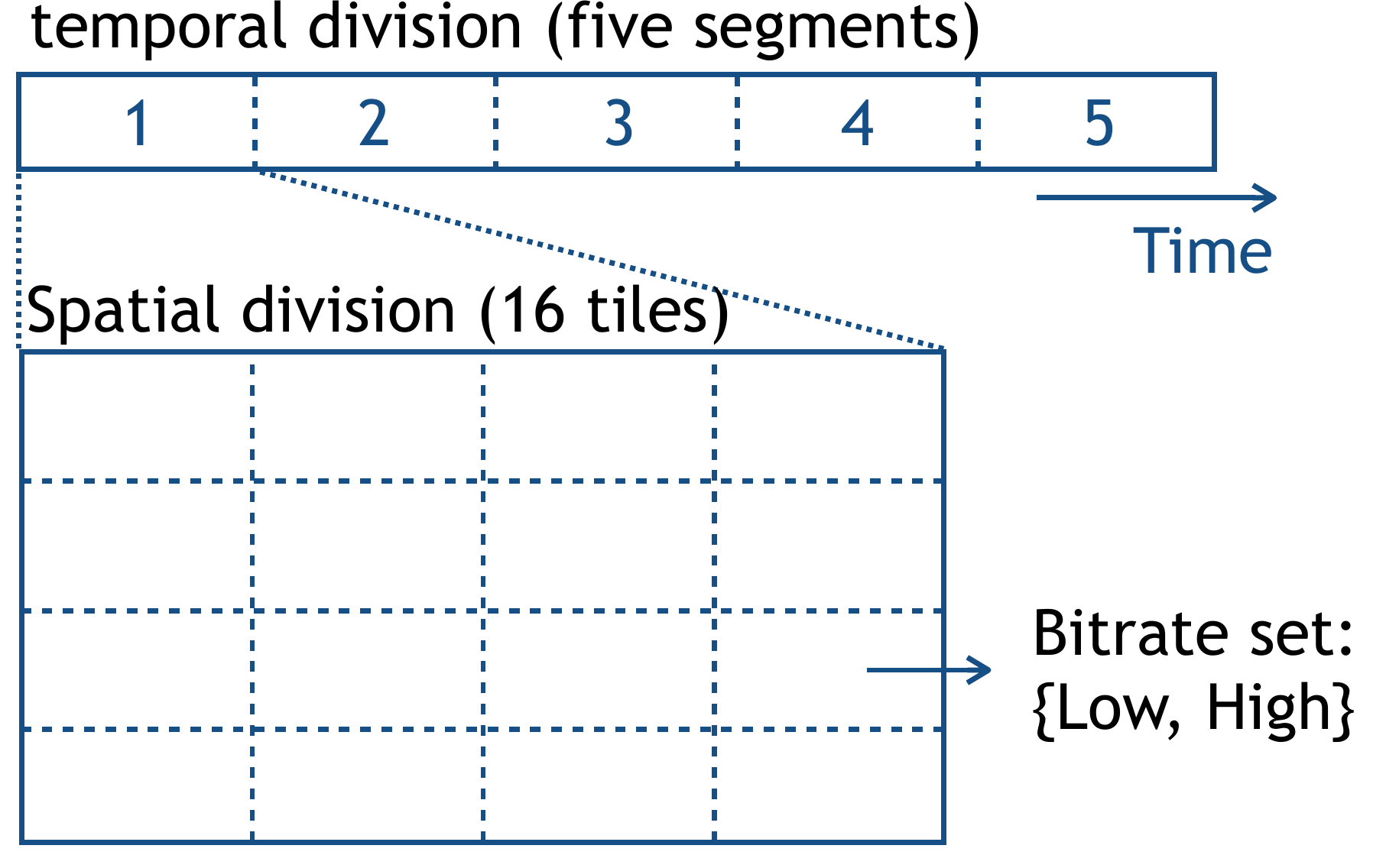}~~\includegraphics[height=2.15cm]{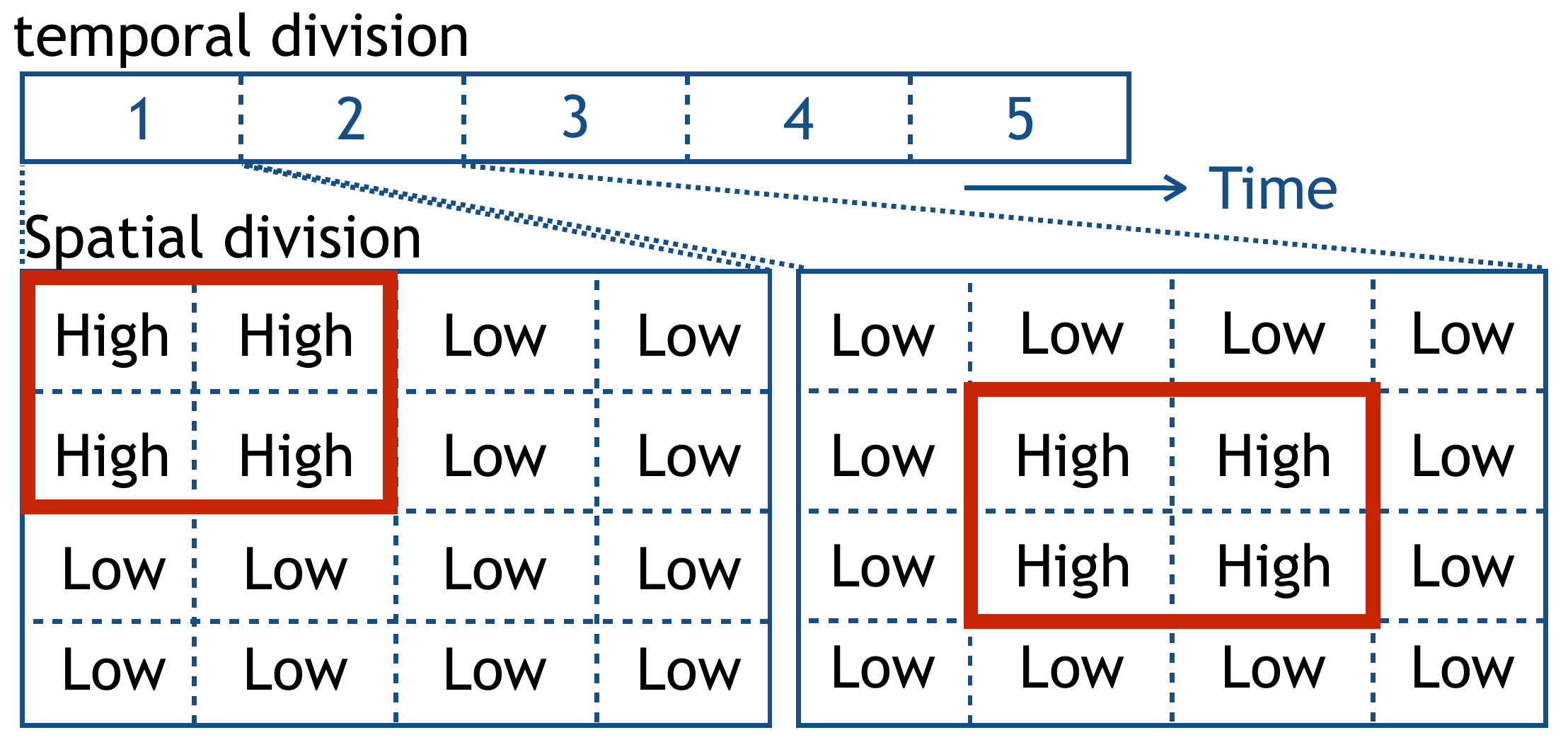}\\
		\small{~~(a)\qquad\qquad\qquad\qquad\qquad\qquad\qquad\qquad(b)\qquad}
		\caption{{Viewport adaptive 360-degree video streaming: (a) temporal and spatial divisions in video encoding; (b) bitrate selection in video streaming.}}\label{fig:model}
	\end{figure}
	To improve the user's QoE, a promising approach is to use viewport adaptive 360-degree (VA360) video streaming \cite{Hosseini2016,DAcunto2016,ballard2019rats}, which exploits bitrate adaptation in both spatial and temporal domains. Specifically, in the encoding process, an entire video is divided into multiple segments, each of which corresponds to a video segment within a certain playback time period. Each segment is further spatially divided into multiple tiles. Each tile corresponds to the video of a viewing area during the segment playback time period. Each tile is encoded at multiple bitrates. During video streaming, video players can select the bitrate of each tile, adapting to human behaviors and real-time network connections. For example, in Fig. \ref{fig:model} (a), in the encoding process, each video is temporally divided into five segments. Each  segment is  spatially divided into $4\times4=16$ tiles, each of which is encoded at two bitrate levels \{Low, High\}. In Fig. \ref{fig:model} (b), when streaming the video,  the video player can select the bitrate of each tile to adapt to the variation of the user's FoV (represented by the solid rectangles) so as to improve the user's QoE and reduce the  required downloading bandwidth. For example, the video player can select high bitrates for the tiles in the FoV and low bitrates for the others. 

	
	In this work, under the VA360 video streaming model, we focus on a video streaming scenario with a newly generated 360-degree video, where no historical viewing information from other users is available for predicting the viewing behaviors of a user. Designing a  bitrate selection algorithm for this scenario is challenging. First, the user's FoV for viewing a segment is unknown when the video player decides the bitrates of the segment, and the FoV may vary across segments. 
	Second,  the actual capacity for downloading a segment is unknown beforehand and may vary across time. 
	These unknown and varying user's FoV and downloading capacity require the bitrate selection algorithm to learn each user's features in an online fashion and adapt the bitrates accordingly.

	\subsection{Related Work}\label{sec:literature}

	Existing works have studied the bitrate selection algorithm design for VA360 video streaming. 
	Using the statistics of other users' viewing history, 
	Xiao \emph{et al.} in \cite{Xiao2018-BAS360} proposed an online bitrate selection algorithm based on the available bandwidth and the probability that a tile is being viewed. Qian \emph{et al.} in \cite{Qian2018-Flare} proposed an algorithm based on the classification of the tiles of each segment using the viewing history. 
	In \cite{zhou2018clustile}, Zhou \emph{et al.} proposed an algorithm to minimize the bandwidth utilization according to the user's expected FoV. Jiang \emph{et al.} in \cite{jiang2019hierarchical} proposed a two-layer bitrate selection algorithm based on both FoV prediction and buffer management.  In \cite{xie2018cls}, Xie \emph{et al.} designed a bitrate selection algorithm using machine learning. In \cite{sun2019two}, Sun \emph{et al.} proposed a two-tier system that selects high quality bitrates for tiles within the predicted FoV and low quality bitrates for the others. These works considered the bitrate selection according to the other users' viewing history, which may not be suitable for those newly generated videos. In \cite{yuan2019spatial}, Yuan \emph{et al.} proposed to use a Gaussian model to predict the user's FoV without using the other users' viewing history. This method, however, requires the video player to define the potential FoV patterns of the user beforehand. Le Feuvre \emph{et al.} in \cite{LeFeuvre2016-Tiled} proposed an algorithm that identifies the objects in each tile and decides the bitrates based on the identifying results. However, this algorithm may not be able to characterize the user heterogeneities. In \cite{shi2019mobile}, Shi \emph{et al.} proposed a head movement-based approach, while its performance depends on the accuracy of  movement prediction. Nguyen in \cite{Nguyen2019-an} proposed an algorithm that adjusts the FoV estimation error to capture the user's FoV preference. However, the proposed algorithm does not consider the bandwidth heterogeneity. Zhang \emph{et al.} in \cite{Zhang2018} proposed a deep reinforcement learning (DRL) algorithm to learn the user's FoV and bandwidth. Pang in \cite{pang2019towards} proposed a DRL algorithm to learn the user's FoV and head movements. The learning in DRL algorithms requires random exploration. For a newly generated video without offline training, the video player learns during the streaming process,  under which the randomly selected bitrates may lead to poor QoE. 
	Online convex optimization \cite{hazan2016introduction,shalev2012online} is a technique which can learn the users' features in an online fashion. 
	Typical methods include online gradient descent (OGD)  algorithms \cite{Mokhtari2016-online,Chen2017-optimization,paternain2017online}  and  online mirror descent (OMD)  algorithms \cite{Hall2015online,jadbabaie2015online}. 
	One of the common characteristics of these algorithms is that the decision of an object (e.g., a segment in our work) is made based on the decision of the previous object and the realization of the  parameters related to the previous object. 
	Such algorithms cannot be directly used in the VA360 video streaming model for the following reason. The realization of the parameters related to a segment  (e.g., the capacity for downloading the segment, the FoV for viewing the segment) is observed after the segment has been downloaded and viewed. When the video player needs to make the bitrate decision of a segment, however, the user may not have viewed the previous segment. 
	In this case,  the video player is unable to determine the bitrates of the segment by applying those existing algorithms  \cite{Mokhtari2016-online,Chen2017-optimization,paternain2017online,Hall2015online,jadbabaie2015online}. 

	

	\subsection{Solution Approach and Contributions}
	In this work, we propose  an online bitrate selection algorithm, called OBS360 algorithm. This  algorithm aims to optimize the user's QoE in VA360 video streaming. It  can  learn the user's FoV preference and time-varying downloading  capacity in real time, and can handle the uncertain FoV and  downloading capacity of the user in the future. 
	
	
	
	The idea of the OBS360 algorithm is inspired by the existing OGD algorithms  \cite{Mokhtari2016-online,Chen2017-optimization,paternain2017online}.  
	In contrast with the existing OGD algorithms, our proposed algorithm can handle the scenario in which the video player may not have observed the realization of the time-varying downloading capacity as well as the user's FoV of the previous segment when it needs to make the bitrate decision of a segment. In addition, under the special case that the realization of the parameters related to a decision can  always be observed before the next decision, the proposed algorithm is equivalent to the existing OGD algorithms.

	The main contributions of this work are as follows.
	\begin{itemize}
		\item \emph{Viewport-Adaptive 360-Degree Video Streaming:} We focus on a 360-degree streaming scenario with a newly generated video, where the video does  not have historical viewing information from other users. Under this scenario, we formulate a bitrate selection problem, which aims to optimize the  user's QoE. 
		\item \emph{Online Bitrate Selection Algorithm:} 
		We propose an OBS360 algorithm for online bitrate selection. The proposed algorithm can learn the user's FoV preference and downloading  capacity in real time. It can achieve sublinear dynamic regret under a convex decision set. This shows the intuition that as the number of segments increases, the performance of the online algorithm approaches the offline optimal performance where the user's FoV and downloading capacity are known beforehand.
		\item \emph{Performance Evaluation:} 
		Simulations with the real-world  datasets from \cite{van2016http} and \cite{Knorr2018data} show the following. The proposed OBS360 algorithm can significantly improve the user's QoE when compared with  BAS-360$^{\circ}$ in \cite{Xiao2018-BAS360} and Flare in \cite{Qian2018-Flare}. Specifically, our proposed algorithm can improve the user's viewing bitrate by $24.6\%-58.8\%$. In addition, it can reduce the inter-segment and intra-segment degradation losses by $83.3\%-91.0\%$ and $84.1\%-89.1\%$, respectively. 
	\end{itemize}
	
	The rest of this paper is organized as follows. We present the system model in Section \ref{sec:system-model}. In Section \ref{sec:online}, we propose the online bitrate selection algorithm. In Section \ref{sec:simulation}, we show the  performance evaluation. We conclude in Section \ref{sec:conclude}.
	

\section{System Model}\label{sec:system-model}
In this section, we first introduce the model setting, and then formulate the user's QoE maximization problem.

\subsection{Model Setting}
We focus on the bitrate adaptation of a user's 360-degree video streaming. The video and user models are as follows.

\subsubsection{Video Model}\label{subsec:video-model}
The video is temporally partitioned into segments (i.e., small video pieces), each corresponding to a playback time of $\beta$ seconds. Let  $\mathcal{I}=\{1,2,\ldots,I\}$ denote the set of segments. Each segment is further spatially divided into $K$ tiles with $M$ rows and $N$ columns, i.e., $K=MN$. Let $\mathcal{K}=\{1,2,\ldots,K\}$ denote the set of tiles of each segment.

We introduce a \emph{reference FoV} for each segment, where the idea is inspired by the recommended FoV in \cite{Knorr2018data}. In practice, the reference FoV can be either the FoV showing the main object (determined by object detection methods, e.g., \cite{LeFeuvre2016-Tiled}) of the video scene or the recommended FoV  marked by video producers (e.g., \cite{Knorr2018data}). The reference FoV is used to characterize user-specific \emph{FoV preference}, i.e., the relative position of a user's FoV to the reference FoV. Fig. \ref{fig:vr-example} shows an example of a 360-degree video streaming of car riding experience. The reference FoV can be looking ahead in the car (e.g., the four shaded tiles in the center) in Fig. \ref{fig:vr-example}, while a user may prefer to look toward his  right-hand side (e.g., the FoV represented by the solid rectangle) in Fig. \ref{fig:vr-example}. 

The reference FoV can be different in different segments. Hence, in order to characterize the user's FoV preference (i.e., the relative position of the user's FoV to the reference FoV), 
we define the indices of the tiles based on their relative positions to the reference FoV.\footnote{If using a fixed tile indexing (irrelevant to the reference FoV) as in most of the existing works in VA360 (e.g., \cite{Qian2018-Flare,Xiao2018-BAS360}), we cannot characterize user's FoV preference because the reference FoV can vary across segments. This is one approach of tile indexing, while it does not affect either the video encoding or the video streaming process.} 
The indices of the tiles are defined as follows. Suppose the top-left corner of the reference FoV  is at the tile on row $m_0$ and  column $n_0$. The tile on row $m$ and column   $n$ is indexed with $ N \text{mod}(m-m_0,M)  + (\text{mod}(n-n_0,N) + 1) $, where $\text{mod}(x,y)$ is equal to $x$ modulo $y$, and $M$ and $N$ is the total number of rows and columns of the tiles, respectively.
For example, in Fig. \ref{fig:index-example}, the reference FoV is represented by the shaded area, which can be different for different segments. In segment $1$, the top-left corner of the reference FoV is located at the tile on row  $m_0 = 2$ and  column $n_0=2$, so the  tile on row $m = 2$ and column $n=4$ is indexed by $4\text{mod}(2-2,4) + (\text{mod}(4-2,4) + 1) = 3$. Note that the 360-degree video has no boundary, e.g., in segment $1$, tile $4$ is adjacent to the right-hand side of tile $3$.  



\begin{figure}[t]
	\centering
	\includegraphics[height=2.5cm]{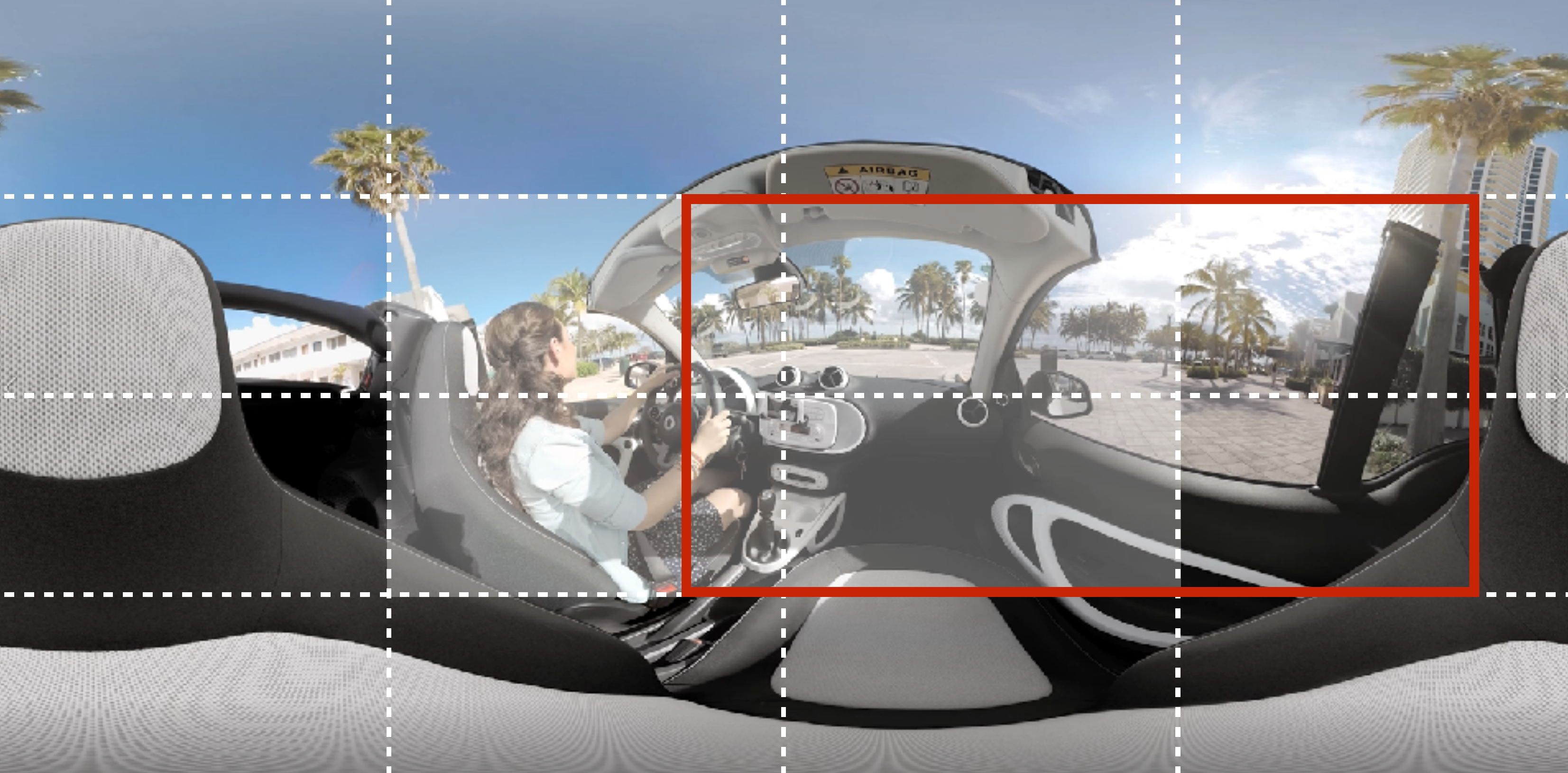}
	\caption{{An example with a 360-degree car riding video.}}\label{fig:vr-example}
\end{figure}

\begin{figure}[t]
	\centering
	\includegraphics[height=2.7cm]{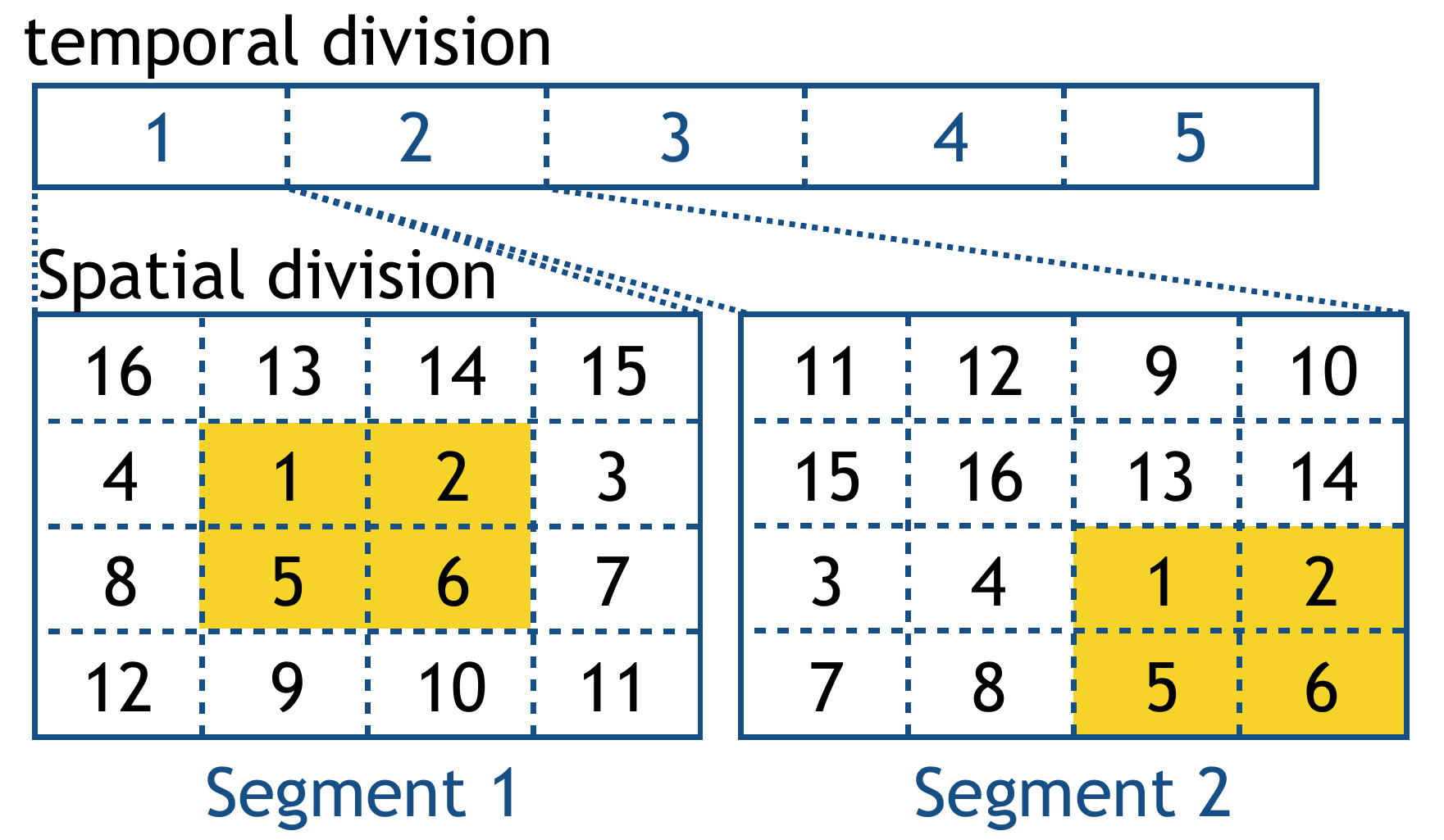}
	\caption{{An illustration of the reference tiles and tile indices.}}\label{fig:index-example}
\end{figure}

In the following, we use tile $(i,k)$ to denote  tile $k$ of segment $i$. 
Each tile is encoded at $|\mathcal{R}|$ bitrate levels, where each level is selected from bitrate set $\mathcal{R} = \{R_1,R_2,...,R_{|\mathcal{R}|}\}$. 
The value of the bitrate represents the number of bits required to encode a tile corresponding to a playback time period of one second. 
Hence, the size (i.e., the total number of bits) of a tile of a $\beta$-second segment selecting bitrate $r\in\mathcal{R}$ is $r \beta$. 
Without  loss of generality, we assume that $R_1< R_2< \cdots< R_{|\mathcal{R}|}$.

\subsubsection{User Model}\label{subsec:user-model}
During  video streaming, the user's downloading capacity and his FoV are time-varying. The user's downloading capacity varies across time. We consider a continuous time interval $\mathcal{T}=[0,T)$. Let $d(t)$ denote the user's downloading capacity at time $t\in\mathcal{T}$, i.e., the user's maximum achievable downloading rate at time $t$. We assume that $d(t)$ is upper- and lower-bounded, i.e., $d(t)\in[d^{\text{min}},d^{\text{max}}]$ for all $t\in\mathcal{T}$. On the other hand, the user's FoV varies across playback time (instead of the real time),\footnote{For example, a user starts playing the video at time $t=0$ sec, and the video rebuffers for one second during time interval $[1,2]$ (sec). Then, at real time $t=3$ sec, the playback time is at $3-(2-1)=2$ sec.} because the user changes his FoV to watch his interested contents in the video, which is playback time-associated. We assume that when the user watches a segment, the user's FoV is unchanged. This is reasonable because in practice, a segment always corresponds to a video piece of several seconds. 

\subsection{Problem Formulation}\label{subsec:problem-formulation}
We aim to optimize the bitrate selection of a user to maximize the user's QoE. 
In the following, we first introduce the decision variables. We then describe the constraints and the QoE. Finally, we formulate the QoE maximization problem.

\subsubsection{Decision Variables}
The user makes decisions on the bitrates of the tiles. Let $r_{i,k}\in \mathcal{R}$ denote the bitrate decision of  tile $(i,k)$. 
Let  $\tau_{i,k}\in\mathcal{T}$  and $\hat{\tau}_{i,k}\in\mathcal{T}$ denote the time that the downloading of tile $(i,k)$ is started and finished, respectively. 
The user downloads the video tiles in a particular sequence: tile $(i,k)$ is downloaded earlier than tile $(i',k')$ if and only if either (a) $i<i'$ or (b) $i=i'$ and $k<k'$. When  one tile has been downloaded, the downloading of the next tile will start immediately, i.e.,
\begin{equation}\label{eq:download}
	{\tau}_{i,k}=\left\{
	\begin{array}{ll}
		\hat{\tau}_{i,k-1},& k>1,\\
		\hat{\tau}_{i-1,K},& k=1.
	\end{array}
	\right.
\end{equation} 
Let $b_{i}$ denote the buffer occupancy when all the tiles of segment $i$ have been downloaded. We define the buffer occupancy with respect to the segment index for the presentation simplicity of buffer update (i.e., the buffer is being updated when a segment has been received). The buffer occupancy is in the unit of playback time, which is commonly used in the existing literature on dynamic adaptive streaming over HTTP (DASH), e.g., \cite{Huang2014BAR}. Receiving one segment will lead to a buffer occupancy increased by $\beta$ seconds. 

Without loss of generality, we set $b_0={b}^{\textsc{ini}}$ as the initial buffer occupancy, and  $\hat{\tau}_{0,K}={\tau}_{1,1}=0$ as the initial starting time. We denote the following decision vectors: $\boldsymbol{r}=(r_{i,k}, i\in\mathcal{I}, k\in\mathcal{K})$, $\boldsymbol{r}_i=(r_{i,k}, k\in\mathcal{K})$,  $\boldsymbol{\tau}=(\tau_{i,k}, i\in\mathcal{I}, k\in\mathcal{K})$, $\hat{\boldsymbol{\tau}}=(\hat{\tau}_{i,k}, i\in\mathcal{I}, k\in\mathcal{K})$,  and $\boldsymbol{b}=(b_{i}, i\in\mathcal{I})$. 

\subsubsection{Downloading Capacity and Buffer Update Constraints} 

Within the downloading period of tile $(i,k)$, the capacity constraint ensures that the total size of the downloaded tile should be no larger than the downloading capacity within the downloading period, i.e.,
\begin{equation}\label{eq:cap-constraint}
	r_{i,k}\beta \leq \int_{{\tau}_{i,k}}^{\hat{\tau}_{i,k}} d(t) dt,~   i\in\mathcal{I}, k\in\mathcal{K}.
\end{equation}

When all the tiles of segment $i$ have been received, the buffer is updated as follows:
\begin{equation}\label{eq:buffer-constraint}
	b_i = \left[b_{i-1} - \left(\hat{\tau}_{i,K} -{\tau}_{i,1} \right)\right]^{+} + \beta, ~ i\in\mathcal{I},
\end{equation}
where the operator $[x]^{+}=\max\{x,0\}$. In \eqref{eq:buffer-constraint}, if the  buffer occupancy $b_{i-1}$ is no smaller than the downloading period, then buffer occupancy $b_{i}$ is the sum of the buffer occupancy immediately before receiving segment $i$ and the received segment length $\beta$. If the buffer becomes empty before receiving segment $i$, then buffer occupancy $b_i$ is equal to $\beta$.

\subsubsection{User's QoE}\label{subsec:qoe}
For the definition of the user's QoE, let $\omega_{i,k}$ denote the fraction of tile $(i,k)$ that is overlapped with the user's FoV. We  define a function $\mu_i(\boldsymbol{r}_i)$ for any segment decision vector $\boldsymbol{r}_i$, indicating the viewing bitrate when the user views segment $i$ under the bitrate decision, i.e.,
\begin{equation}
	\mu_i(\boldsymbol{r}_{i}) = \sum_{k\in\mathcal{K}}  \omega_{i,k} r_{i,k}, ~ i\in\mathcal{I}.
\end{equation}
Note that this is the bitrate that the user actually views, taking into account  the user's FoV.

Similar to some of the existing works \cite{Qian2018-Flare,Zhang2018}, the user's QoE consists of three terms, which are the user's viewing utility $U(\boldsymbol{r})$, the rebuffering loss $L^{\textsc{rb}}(\boldsymbol{\tau},\hat{\boldsymbol{\tau}},\boldsymbol{b})$, and the bitrate degradation loss $L^{\textsc{bd}}(\boldsymbol{r})$, i.e.,
\begin{equation}\label{eq:qoe}
	Q(\boldsymbol{r},\boldsymbol{\tau},\hat{\boldsymbol{\tau}},\boldsymbol{b}) = U(\boldsymbol{r}) - L^{\textsc{rb}}(\boldsymbol{\tau},\hat{\boldsymbol{\tau}},\boldsymbol{b}) -  L^{\textsc{bd}}(\boldsymbol{r}).
\end{equation}
The user's viewing utility of each segment is a function of the user's  viewing bitrate of the segment, i.e., a larger viewing bitrate leads to a higher utility. The user's viewing utility $U(\boldsymbol{r})$ is the summation of the user's viewing utilities of all segments:
\begin{equation}
	U(\boldsymbol{r})  =  \sum_{i\in\mathcal{I}}u_i(\mu_i(\boldsymbol{r}_{i})),
\end{equation}
where $u_i(\cdot)$ for segment $i$ is a nondecreasing concave function. 
This is a generalization of those works considering linear viewing utility, e.g., \cite{Qian2018-Flare,Xiao2018-BAS360,Zhang2018}.

If the tiles of a segment have not been completely received by the video player when the segment is to be played, then rebuffering will occur, and the video will freeze until the tiles of the segment are received.
The rebuffering loss is the  user's loss resulting from the video freeze. This loss is proportional to the rebuffering time. Hence, the rebuffering loss is defined as follows:
\begin{equation}\label{eq:rebuffering-constraint-v2}
	L^{\textsc{rb}}(\boldsymbol{\tau},\hat{\boldsymbol{\tau}},\boldsymbol{b})  = l^{\textsc{rb}} \sum_{i\in\mathcal{I}} \left[\hat{\tau}_{i,K} - \tau_{i,1} - b_{i-1} \right]^+,
\end{equation} 
where $ l^{\textsc{rb}}$ is the unit loss if the user experiences a one-second rebuffering. In \eqref{eq:rebuffering-constraint-v2}, if the downloading duration of segment $i$ is larger than the buffer occupancy $b_{i-1}$, then rebuffering will happen, leading to a rebuffering time of $\left[\hat{\tau}_{i,K} - \tau_{i,1} - b_{i-1} \right]^+$.


The bitrate degradation loss consists of two parts: an inter-segment bitrate degradation loss $ L^{\textsc{bd-e}}(\boldsymbol{r}) $ and an intra-segment bitrate degradation loss $L^{\textsc{bd-a}}(\boldsymbol{r})$. That is,
\begin{equation}\label{eq:bitrate-degradation}
	L^{\textsc{bd}}(\boldsymbol{r})   = L^{\textsc{bd-e}}(\boldsymbol{r})  +  L^{\textsc{bd-a}}(\boldsymbol{r}).
\end{equation}
The inter-segment bitrate degradation loss is the loss resulting from the  bitrate degradation among segments,  i.e.,
\begin{equation}\label{eq:bitrate-degradation-inter}
	L^{\textsc{bd-e}}(\boldsymbol{r})  = l^{\textsc{bd-e}} \sum_{i\in\mathcal{I}/\{1\}}\left[\mu_{i-1}(\boldsymbol{r}_{i-1}) - \mu_i(\boldsymbol{r}_{i})\right]^{+},
\end{equation}
where $l^{\textsc{bd-e}} $ is the loss if the viewing bitrate is degraded by one unit. In \eqref{eq:bitrate-degradation-inter}, if $\mu_{i-1}(\boldsymbol{r}_{i-1}) $ is larger than $\mu_{i}(\boldsymbol{r}_{i}) $, then the inter-segment bitrate degradation occurs, leading to a degradation loss which is proportional to the size of the degradation. The intra-segment bitrate degradation loss is the loss resulting from the bitrate difference among the tiles of each segment, i.e., 
\begin{equation}\label{eq:bitrate-degradation-intra}
	L^{\textsc{bd-a}}(\boldsymbol{r})  =  l^{\textsc{bd-a}} \sum_{i\in\mathcal{I}}\sum_{k\in\mathcal{K}}\omega_{i,k}\left[\frac{\mu_i(\boldsymbol{r}_{i})}{\sum_{k'\in\mathcal{K}}\omega_{i,k'}}-r_{i,k}\right]^+, 
\end{equation}
where $l^{\textsc{bd-a}} $ is the loss if the bitrate is degraded by one unit. In \eqref{eq:bitrate-degradation-intra}, if $\omega_{i,k}$ is positive (i.e., the tile is being viewed), and the bitrate $r_{i,k}$ is smaller than the normalized viewing bitrate $\mu_i(\boldsymbol{r}_{i})/(\sum_{k'\in\mathcal{K}}\omega_{i,k'})$, then the intra-segment bitrate degradation happens, inducing a degradation loss proportional to the size of the degradation. 

\subsubsection{Problem Formulation} 
We aim to determine the  decision vectors $\boldsymbol{r}$, $\boldsymbol{\tau}$, $\hat{ \boldsymbol{\tau}}$, $\boldsymbol{b}$ to maximize the user's QoE subject to the  capacity and buffer update constraints. The problem is formulated as follows:
\begin{subequations}\label{eq:optimization}
	\begin{align} 
		\displaystyle \mathop \textrm{maximize}\limits_{\boldsymbol{r}, \boldsymbol{\tau}, \hat{ \boldsymbol{\tau}}, \boldsymbol{b}} & ~~Q(\boldsymbol{r},\hat{\boldsymbol{\tau}},\boldsymbol{\tau},\boldsymbol{b}) \\
		\textrm{subject to} & ~~ r_{i,k}\in\mathcal{R}, ~i\in\mathcal{I}, k\in\mathcal{K},\\
		& ~~ 0\leq \tau_{i,k}<T, ~i\in\mathcal{I}, k\in\mathcal{K},\\
		&~~ 0\leq \hat{\tau}_{i,k}<T, ~i\in\mathcal{I}, k\in\mathcal{K},\\
		& ~~ b_{i}\geq 0, ~i\in\mathcal{I},\\
		&~~\text{constraints }\eqref{eq:download}, \eqref{eq:cap-constraint},\eqref{eq:buffer-constraint}. \nonumber
	\end{align}
\end{subequations}
Problem \eqref{eq:optimization} is a mixed-integer nonconvex programming problem with nonlinear constraints \eqref{eq:cap-constraint} and \eqref{eq:buffer-constraint}. This problem is challenging to solve, even in an offline case when all the parameters (i.e., the user's FoV and downloading capacity) are known beforehand. 

In this work, we focus on the online algorithm design, which addresses the realistic scenario where the user's FoV of a segment and the capacity for downloading the segment are unknown when the bitrate decision of the segment needs to be made. 

\section{Online Algorithm Design}\label{sec:online}
In this section, we focus on the online algorithm design. For the design of the online algorithm, we consider a set of per-segment problems, each corresponding to one of the segments of the video. The bitrate decision of a segment will be made based on two kinds of bitrate decisions of a set of previous segments. The first is the \emph{actual} bitrate decisions of the set of previous segments, which are the decisions made according to the proposed algorithm. The second is the \emph{optimal} bitrate decisions of the set of previous segments, which are the decisions obtained by solving the corresponding per-segment problems according to the realizations of the real-time downloading capacities and FoVs of the segments.
The performance of the  algorithm will be characterized by dynamic regret\cite{Chen2017-optimization,Mokhtari2016-online}, reflecting the regret of the algorithm in the objective value. We show that our online algorithm yields sublinear dynamic regret, i.e., as the number of segments increases, the dynamic regret of the online algorithm approaches zero on the long-term average.  

We first introduce the set of per-segment problems and the performance metric. Then, we introduce the online algorithm. After that, we show the performance guarantee of the proposed algorithm under particular conditions. Finally, we modify the online algorithm to address the scenario when the conditions are not satisfied in practice.

\subsection{Per-Segment Problem}
For the design of the online algorithm, we consider a set of per-segment optimization problems. The per-segment problem corresponding to segment $i\in\mathcal{I}$  aims to optimize the bitrate decision of segment $i$ at the time that the bitrate decisions of all the previous segments have been made. 

\subsubsection{Per-Segment Objective Function}
For segment $i\in\mathcal{I}$, we define a function $\widetilde{Q}_{i}(\boldsymbol{r}_i~|~\boldsymbol{r}_{i-1},b_{i-1})$, which indicates the user's QoE of segment $i$ under bitrate decision vector $\boldsymbol{r}_i$, given the bitrate decision vector of segment $i-1$ (i.e., $\boldsymbol{r}_{i-1}$) and the buffer occupancy before downloading segment $i$ (i.e., $b_{i-1}$). The function $\widetilde{Q}_{i}(\boldsymbol{r}_i~|~\boldsymbol{r}_{i-1},b_{i-1})$ is given as follows:
\begin{multline}\label{eq:per-segment-utility}
	\widetilde{Q}_{i}(\boldsymbol{r}_i~|~\boldsymbol{r}_{i-1},b_{i-1}) \\ = u_i(\mu_{i}(\boldsymbol{r}_i)) - l^{\textsc{rb}}\left(\frac{1}{\bar{d}_i}\sum_{k\in\mathcal{K}}r_{i,k}\beta - b_{i-1}\right) \\-l^{\textsc{bd-a}} \sum_{k\in\mathcal{K}}\omega_{i,k}\left[\frac{\mu_{i}(\boldsymbol{r}_i)}{\sum_{k'\in\mathcal{K}}\omega_{i,k'}}-r_{i,k}\right]^+\\-l^{\textsc{bd-e}} \left(\mu_{i-1}(\boldsymbol{r}_{i-1})-\mu_{i}(\boldsymbol{r}_i)\right) ,
\end{multline}
where $\bar{d}_i$ is the average downloading capacity when segment $i$ is being downloaded. For simplification, we will use $\widetilde{Q}_{i}(\boldsymbol{r}_i)$ to denote $\widetilde{Q}_{i}(\boldsymbol{r}_i~|~\boldsymbol{r}_{i-1},b_{i-1})$ in the rest of this paper.

In contrast to the user's QoE for segment $i$ (as in Section \ref{subsec:qoe}), equation \eqref{eq:per-segment-utility} has two major differences. First, the downloading time of segment $i$ in \eqref{eq:per-segment-utility}, i.e., ${\sum_{k\in\mathcal{K}}r_{i,k}\beta}/\bar{d}_i$, is an approximation of the actual downloading time (as it is computed in  \eqref{eq:cap-constraint}). This approximation is introduced to simplify the algorithm design, and is close to the actual downloading time, because the change of the downloading capacity  within the downloading period of a segment is small in practice.\footnote{For example, downloading a two-second 4K segment (which has a bitrate of around $2\times50$ Mbps on YouTube) using a 62 Mbps downlink (as the average mobile download speed in Canada in 2019 \cite{speedtest-canada}) leads to a downloading period of $2\times 50  / 62 \approx 1.6$ seconds.} Second, in \eqref{eq:per-segment-utility}, we relax the operators $[\cdot]^+$ involved in the rebuffering loss and the inter-segment bitrate degradation loss. Intuitively, as we focus on per-segment problems,  this relaxation can help with characterizing the impact of the bitrate decision of a segment on the user's QoE of the subsequent segments. For example, the rebuffering loss may fail to capture the difference between downloading a segment (with smaller bitrates) for one second and downloading a segment (with  larger bitrates) for five seconds, if neither downloading process induces rebuffering. The two downloading processes, however, may lead to different likelihood of having a rebuffering later due to the resulting different buffer occupancies after the downloading has been accomplished. This difference can be characterized if the operator $[\cdot]^+$ is relaxed. A similar reason applies for the inter-segment bitrate degradation loss. 


\subsubsection{Per-Segment Optimization Problem} 
Given $\boldsymbol{r}_{i-1}$ and $b_{i-1}$, the per-segment problem for segment $i$ is to determine the bitrate decision  $\boldsymbol{r}_i$  by  maximizing $\widetilde{Q}_i(\boldsymbol{r}_i)$, i.e., 
\begin{equation*}
	\begin{array}{rl}
		\mathop \textrm{maximize}\limits_{\boldsymbol{r}_i}&~~ \widetilde{Q}_i(\boldsymbol{r}_i)\\
		\textrm{subject to}& ~~ r_{i,k}\in{\mathcal{R}}, ~k\in\mathcal{K}.
	\end{array}\eqno{\text{(OPT-SEGMENT-i)}}
\end{equation*}		
Problem (OPT-SEGMENT-i) does not include constraints \eqref{eq:download}, \eqref{eq:cap-constraint}, and \eqref{eq:buffer-constraint} (as in problem \eqref{eq:optimization}) for the following reasons. Constraint \eqref{eq:download} is eliminated, because under the per-segment problem,  the relationship between the downloading time of different segments does not need to be considered. Constraint \eqref{eq:cap-constraint} is eliminated due to the downloading time approximation. In addition, with the set of per-segment problems, at the time that the bitrate decision of segment $i$ is optimized, the buffer occupancy $b_{i-1}$ is directly observed, so the buffer update constraint \eqref{eq:buffer-constraint}  is not required.

\subsection{Performance Metrics}\label{subsec:metrics}
Let $\boldsymbol{r}_{i}^{o}$ denote the bitrate decision of segment $i$ resulting from the online algorithm. We evaluate the performance of  the online algorithm using dynamic regret \cite{Chen2017-optimization,Mokhtari2016-online}. The dynamic regret is defined as the difference between the user's QoE under the optimal bitrate selection decision (i.e., the decision after the realization of the user's downloading capacity and FoV) and  the user's QoE under the actual bitrate selection decision determined by the proposed algorithm (i.e., the decision before the realization of the user's downloading capacity and FoV). The definition of the dynamic regret is as follows:
\begin{equation}\label{eq:dynamic-regret}
	\text{Reg}_{I} \triangleq \sum_{i\in\mathcal{I}}\left(\widetilde{Q}_i(\boldsymbol{r}_i^*) - \widetilde{Q}_i(\boldsymbol{r}_{i}^{o})\right),
\end{equation}
where $\boldsymbol{r}_i^*$ is the optimal bitrate decision for segment $i$ obtained by solving problem (OPT-SEGMENT-i). We aim to design an online algorithm that achieves sublinear  dynamic regret, i.e., 
\begin{equation}
	\lim_{I\rightarrow\infty} \text{Reg}_{I} / I = 0, 
\end{equation}
which implies that as the number of segments  approaches infinity, the dynamic regret of the online algorithm  on the long-term average (i.e., $\text{Reg}_{I} / I$) approaches zero.


\subsection{Online Bitrate Selection Algorithm}\label{subsec:algorithm}
We design an online algorithm, called OBS360 algorithm. 
The online algorithm design is inspired by the OGD algorithms, e.g.,  \cite{Mokhtari2016-online,Chen2017-optimization,paternain2017online}. 
Those existing algorithms \cite{Mokhtari2016-online,Chen2017-optimization,paternain2017online}, however, cannot be directly applied in VA360 video streaming service. This is because for those algorithms, the decision of an object (e.g., a segment) is made based on the decision of the  previous object and the realization of the  parameters related to that object. In VA360 video streaming, however, the parameter realization of the previous object may not have been observed when the decision of an object is to be made. To address this, in our proposed OBS360 algorithm, a decision is made based on the decisions of a set of objects whose parameter realizations have been observed and  the corresponding parameter realizations of those objects. Such an algorithm design is more challenging than those existing ones, because the decision making of an object should take into account the parameter realizations and decisions of a set of objects as well as the parameter realization correlations among the objects.  Under a special case that at the time when a decision is made, the parameter realization of the  previous decision can always be obtained, our proposed algorithm is equivalent to the existing OGD algorithms \cite{Mokhtari2016-online,Chen2017-optimization,paternain2017online}.

In the following, we first introduce an auxiliary set for each segment, which indicates the segments whose parameter  realizations (e.g., capacity for downloading the segment, FoV for viewing the segment) are observed after the bitrate decision of the previous segment is made. Then, we present our proposed  OBS360 algorithm.


\subsubsection{Auxiliary Set}\label{eq:j}
Let $\tilde{I}_i$ denote the segment index such that segment $i$ has been viewed by the user during the time period when the bitrate decisions of segments $\tilde{I}_i$ and $\tilde{I}_i+1$ are made.
Since a segment should be downloaded before it is being viewed, the FoV and downloading capacity information of segment $i$ has been observed by the video player when the bitrate decision of segment $\tilde{I}_i + 1$ is to be made.   
The auxiliary set $\mathcal{I}_i\subset\mathcal{I}$ for  segment $i\in\mathcal{I}\cup\{I+1\}$ is defined as follows:
\begin{equation}\label{eq:ii}
	\mathcal{I}_i\triangleq\{i'\in\mathcal{I}~|~\tilde{I}_{i'} = i -1 \},
\end{equation}
which is the set of segments which have been viewed by the user during the time period when the bitrate decisions of segments $i -1$ and $i$ are made. The consideration of $i=I+1$ is used to include the segments whose parameter realizations are observed after the bitrate decision of segment $I$ is made.\footnote{Note that given the bitrate decisions and the parameter realizations of all segments in $\mathcal{I}$, the vectors $\mathcal{I}_{I+1}$ and $\boldsymbol{r}^o_{I+1}$ can be determined.}


\subsubsection{OBS360 Algorithm}
The OBS360 algorithm is given in Algorithm \ref{alg:saddle-point}. The algorithm first initializes the following parameters: an initial bitrate vector $\boldsymbol{r}_0$, which can be any bitrate in set $\mathcal{R}^K$; a parameter $\alpha>0$, which corresponds to the positive stepsize in the existing OGD algorithms \cite{Mokhtari2016-online,Chen2017-optimization,paternain2017online}.

\begin{figure}[t]
	\begin{minipage}{\linewidth}
		\begin{algorithm}[H]
			\caption{OBS360 Algorithm}\label{alg:saddle-point}
			\small
			\begin{algorithmic}[1]
				\STATE \textbf{Initialization:} $\boldsymbol{r}_0$ and  $\alpha$;
				\FORALL{$i\in\mathcal{I}$}
				\STATE{Obtain $\mathcal{I}_i$ according to \eqref{eq:ii};}
				\IF{$\mathcal{I}_i\neq \emptyset$}
				\STATE{Compute $J_i$ according to \eqref{eq:online-J};}
				\STATE{Update $\boldsymbol{r}^o_i$ according to \eqref{eq:online-update};}
				\ELSE
				\STATE{Update $\boldsymbol{r}^o_i$ according to \eqref{eq:online-update-same};}
				\ENDIF
				\ENDFOR
			\end{algorithmic}
		\end{algorithm}
	\end{minipage}
\end{figure}

For any segment $i$, the video player first obtains the auxiliary set $\mathcal{I}_i$, which contains the indices of the segments whose FoV and downloading capacity information is newly observed as defined in \eqref{eq:ii}. If the set $\mathcal{I}_i$ is non-empty, then the bitrate decision of segment $i$ will be made based on the bitrate decisions of the segments referred in set $\mathcal{I}_i$ and the corresponding downloading capacity and FoV information when each of those segments is downloaded and viewed, respectively. Specifically, the bitrate decision of segment $i$  is updated as follows:
\begin{equation}\label{eq:online-update}
	\boldsymbol{r}_i^o =  \arg \min_{\boldsymbol{r}_i\in{\mathcal{R}}^K} - \nabla \widetilde{Q}_{J_i}(\boldsymbol{r}_{J_i}^o)^{\top} (\boldsymbol{r}_i-\boldsymbol{r}_{J_i}^o) + \frac{1}{2\alpha} \lVert \boldsymbol{r}_i - \boldsymbol{r}_{J_i}^o\rVert^2,
\end{equation}
where $\nabla \widetilde{Q}_{J_i}(\boldsymbol{r}_{J_i}^o)$ is the subgradient of $\widetilde{Q}_{J_i}(\boldsymbol{r}_{J_i}^o)$, and  $J_i$ is defined as 
\begin{equation}\label{eq:online-J}
	J_i \triangleq \arg \min_{j\in\mathcal{I}_i} - \nabla \widetilde{Q}_j(\boldsymbol{r}_j^o)^{\top} (\boldsymbol{r}_j^*-\boldsymbol{r}_j^o).
\end{equation}
Specifically, $J_i $ is the segment index in set $\mathcal{I}_i$ such that the dot product of the negative value of the subgradient $\nabla \widetilde{Q}_j(\boldsymbol{r}_j^o)$ and the difference between the optimal decision of segment  $J_i $ (obtained by solving the corresponding per-segment problem) and the actual decision of segment  $J_i $ (obtained according to the algorithm) is the smallest, comparing with other segments in  set $\mathcal{I}_i$.
If  the set $\mathcal{I}_i$ is empty, then the bitrate decision of segment $i$  is the same as the bitrate decision of segment $i-1$:
\begin{equation}\label{eq:online-update-same}
	\boldsymbol{r}_i^o = \boldsymbol{r}_{i-1}^o.
\end{equation}

In the special case when set $\mathcal{I}_i = \{i-1\}$ for all $i\in\mathcal{I}\cup\{I+1\}$, i.e., the information of segment $i-1$ can always be observed during the time period  when the decisions of segments $i-1$ and $i$ are made,  Algorithm \ref{alg:saddle-point}  is equivalent to the existing OGD  algorithms in \cite{Mokhtari2016-online,Chen2017-optimization,paternain2017online}.

\subsection{Performance Analysis}\label{subsec:online-performacne}
We proceed to show that Algorithm \ref{alg:saddle-point} leads to sublinear dynamic regret. Note that the analysis in this section is based on a case where the bitrate set $\mathcal{R}$ is convex. This is commonly assumed when the regret is derived in the existing online convex optimization algorithms \cite{hazan2016introduction,shalev2012online, Mokhtari2016-online,Chen2017-optimization,paternain2017online,Hall2015online,jadbabaie2015online}. Hence, our result on sublinear dynamic regret reveals the performance of the proposed algorithm under the particular case, and it provides an insight that as the number of segments increases, the performance of the proposed online algorithm approaches the offline optimal solution. The performance of the proposed algorithm over a practical bitrate set is evaluated using simulations in Section \ref{sec:simulation}.

In the following, we first show a condition regarding the variation of the time-varying parameters, and then present the bound of the dynamic regret of Algorithm \ref{alg:saddle-point}. Finally, we formally state the sublinear dynamic regret.


\subsubsection{Condition on Parameter Variation}
In online convex optimization techniques considering time-varying parameters, without restricting the varying of the parameters, obtaining a bound on dynamic regret is not possible \cite{jadbabaie2015online}. Here, we impose a condition on the varying of the parameters as follows:
\begin{condition}[Time-Varying Parameters]\label{ass:parameters}
	The variation of the time-varying parameters, i.e., downloading capacity $(d(t),t\in\mathcal{T})$ and FoV $(\boldsymbol{\omega}_{i},i\in\mathcal{I})$, should be small in the sense that there exist nonnegative values $V_{\emptyset}$ and $V_{r}$ such that the following inequalities hold for  any number of segments $I$:
	\begin{itemize}
		\item The number of the set  $\mathcal{I}_i$ that satisfies $\mathcal{I}_i=\emptyset$ for all $i\in\mathcal{I}\cup\{I+1\}$ is bounded, i.e.,  
		\begin{equation}\label{eq:assumption-1}
			\sum_{i\in\mathcal{I}\cup\{I+1\}} \mathbbm{1}(\mathcal{I}_i=\emptyset)  \leq V_{\emptyset},
		\end{equation}
		where $\mathbbm{1}(\cdot)$ is an indicator function, i.e., $\mathbbm{1}(\mathcal{I}_i=\emptyset) = 1$ if $\mathcal{I}_i=\emptyset$, and is equal to zero otherwise.
		\item Let $J_{i}^{\dag}$ denote the index of the segment such that the bitrates of segment $J_i$ are made based on it, i.e., $J_{i}^{\dag}=J_{i'}$ with $i'=J_i$.\footnote{If $\mathcal{I}_{i}=\emptyset$, we define $ J_i \triangleq J_{h(i)}$, where $h(i) = \{h~|~\mathcal{I}_{h} \neq \emptyset, \mathcal{I}_{h+1} = \mathcal{I}_{h+2} = \cdots =\mathcal{I}_{i} = \emptyset\}$  is the index of the segment which is the last segment before segment $i$ (i.e., $h(i)< i$) whose $\mathcal{I}_{h(i)}\neq \emptyset$.
		} 
		The difference between the optimal solution to per-segment problem ${J_i}$ and that to per-segment problem ${J_{i}^{\dag}}$ is small for all $i\in\mathcal{I}\cup\{I+1\}$ such that 
		\begin{equation}\label{eq:assumption-3}
			\sum_{i\in\mathcal{I}\cup\{I+1\}}|\mathcal{I}_i|\lVert \boldsymbol{r}^*_{J_i} -  \boldsymbol{r}^*_{J_{i}^{\dag}} \rVert \leq V_{r},
		\end{equation}
		where $|\mathcal{I}_i|$ is the cardinality of  set $\mathcal{I}_i$.
	\end{itemize} 
\end{condition}
In Condition \ref{ass:parameters}, inequality \eqref{eq:assumption-1}  implies that the number of the bitrate decisions that are not updated based on \eqref{eq:online-update} is bounded. Inequality \eqref{eq:assumption-3} ensures that the variations of the parameters are relatively small, such that the changes among  the optimal solutions to the per-segment problems are bounded. In the special case when $\mathcal{I}_i = \{i-1\}$ for all $i\in\mathcal{I}\cup\{I+1\}$, Condition \ref{ass:parameters} is equivalent to each of the parameter variation conditions in the existing works  \cite{Chen2017-optimization,Mokhtari2016-online} on OGD algorithms.


\subsubsection{Bound of the Dynamic Regret}
In the following, we first show the bound of the regret of each segment under Algorithm \ref{alg:saddle-point}. We then show the bound of the dynamic regret.

The bound of the regret of each segment  is given in Lemma \ref{lem:app1}. Note that in this lemma, we use two different subscripts $s$ and $i$ to refer to the indices of different segments in order to make the presentation clear. 
\begin{lemma}[Regret of Each Segment]\label{lem:app1}
	Under Condition \ref{ass:parameters} and convex decision set $\mathcal{R}$, the regret of any segment $s\in\mathcal{I}$, i.e., $\widetilde{Q}_s(\boldsymbol{r}^*_s)-\widetilde{Q}_{s}(\boldsymbol{r}_s^o)$, is upper-bounded as follows:
	\begin{itemize}	
		\item[{(A)}]  For any segment $s\in\mathcal{I}_i$, $i\in\mathcal{I}\cup\{I+1\}$, 
		\begin{multline}\label{eq:app-lemma1}
			\widetilde{Q}_s(\boldsymbol{r}^*_s) -\widetilde{Q}_{s}(\boldsymbol{r}_s^o) \leq \frac{1}{2\alpha} \Big(2R \lVert\boldsymbol{r}_{J_i}^*- \boldsymbol{r}_{J_i^{\dag}}^*\rVert  \\
			+ \lVert\boldsymbol{r}_{J_i}^o - \boldsymbol{r}_{J_i^{\dag}}^*\rVert^2
			-\lVert \boldsymbol{r}_{J_i}^* -\boldsymbol{r}_i^o\rVert^2 \Big) +\frac{\alpha}{2} \overline{Q}^2  .
		\end{multline}
		\item[{(B)}] For any  segment $s\in\widetilde{\mathcal{I}}\triangleq \{s'\in\mathcal{I}~|~s'\notin\mathcal{I}_i,i\in\mathcal{I}\cup\{I+1\}\}$,
		\begin{equation}\label{eq:app-lemma2}
			\widetilde{Q}_s(\boldsymbol{r}^*_s) -\widetilde{Q}_{s}(\boldsymbol{r}_s^o) 
			\leq \frac{3R^2}{2\alpha} +\frac{\alpha}{2} \overline{Q}^2.
		\end{equation}
	\end{itemize}
	The constant $\overline{Q}$ is the bound of the norm of the subgradient  $\nabla \widetilde{Q}_i(\boldsymbol{r}_{i})$, i.e., $\lVert\nabla  \widetilde{Q}_i(\boldsymbol{r}_i)\rVert\leq \overline{Q}$ for all $i\in\mathcal{I}$ and $\boldsymbol{r}_i\in{\mathcal{R}}^K$. The constant $R$ is the radius of the convex set ${\mathcal{R}}^{K}$, i.e., $\lVert \boldsymbol{r}_i - \boldsymbol{r}_j\rVert\leq R$, for all $\boldsymbol{r}_i,\boldsymbol{r}_j\in{\mathcal{R}}^{K}$.  
\end{lemma}
\begin{proof}
	For the proof of the statements (A) and (B) in Lemma \ref{lem:app1}, we first show an inequality that holds for any segment $i\in\mathcal{I}\cup\{I+1\}$ whose $\mathcal{I}_i\neq \emptyset$. Specifically, for any of these segments, the bitrate decision $\boldsymbol{r}_{i}^{o}$ is the optimal solution of the following optimization problem (according to \eqref{eq:online-update}):
	\begin{equation}\label{eq:app-online-update-r}
		\min_{\boldsymbol{r}_i\in{\mathcal{R}}^K} f_i(\boldsymbol{r}_i)\triangleq -\nabla \widetilde{Q}_{J_i}(\boldsymbol{r}_{J_i}^o)^{\top} (\boldsymbol{r}_i-\boldsymbol{r}_{J_i}^o) + \frac{1}{2\alpha} \lVert \boldsymbol{r}_i - \boldsymbol{r}_{J_i}^o\rVert^2.
	\end{equation}
	It can be shown that function $f_i(\boldsymbol{r}_i)$ in \eqref{eq:app-online-update-r} is $1/\alpha$-strongly convex (Section 3.4 in \cite{bubeck2014theory}) for any $\alpha > 0$, and this leads to the following inequality: for any $i\in\mathcal{I}\cup\{I+1\}$ with $\mathcal{I}_i\neq \emptyset$,
	\begin{equation}\label{eq:app-plug}
		f_i(\boldsymbol{r}_{J_i}^*) 
		\geq f_i(\boldsymbol{r}_{i}^o) + \frac{1}{2\alpha} \lVert \boldsymbol{r}_{J_i}^*- \boldsymbol{r}_{i}^o \rVert^2.
	\end{equation}
	
	Based on inequality \eqref{eq:app-plug}, we now prove the statements (A) and (B) in Lemma \ref{lem:app1} as follows.
	
	\emph{Proof for Statement (A):} 
	For  any $s\in\mathcal{I}_i$ with $i\in\mathcal{I}\cup\{I+1\}$, by adding $-\widetilde{Q}_{s}(\boldsymbol{r}_s^o)$ to both sides of the inequality \eqref{eq:app-plug}, substituting  $f_i(\boldsymbol{r}_{J_i}^*)$ and $f_i(\boldsymbol{r}_i^o)$ according to \eqref{eq:app-online-update-r}, and reordering the inequality, we have
	\begin{align}\label{eq:app-k2}
		\begin{split}
			&~-\widetilde{Q}_{s}(\boldsymbol{r}_s^o)  - \nabla \widetilde{Q} _{J_i}(\boldsymbol{r}_{J_i}^o)^{\top} (\boldsymbol{r}_i^o -\boldsymbol{r}_{J_i}^o) \\
			\overset{\text{(a)}}{\leq} &~\frac{1}{2\alpha} \left(\lVert \boldsymbol{r}_{J_i}^*- \boldsymbol{r}_{J_i}^o\rVert^2 -\lVert \boldsymbol{r}_{J_i}^*- \boldsymbol{r}_{i}^o \rVert^2 - \lVert \boldsymbol{r}_{i}^o - \boldsymbol{r}_{J_i}^o\rVert^2 \right)\\ 
			&~- \widetilde{Q}_{s}(\boldsymbol{r}_s^o) -\nabla \widetilde{Q}_{s}(\boldsymbol{r}_{s}^o)^{\top} (\boldsymbol{r}_s^*-\boldsymbol{r}_{s}^o)\\
			\overset{\text{(b)}}{\leq} &~\frac{1}{2\alpha} \left(\lVert \boldsymbol{r}_{J_i}^*\!- \!\boldsymbol{r}_{J_i}^o\rVert^2 -\lVert \boldsymbol{r}_{J_i}^*\!- \! \boldsymbol{r}_{i}^o \rVert^2 - \lVert \boldsymbol{r}_{i}^o \!- \! \boldsymbol{r}_{J_i}^o\rVert^2 \right) \\
			&~-  \widetilde{Q}_s(\boldsymbol{r}^*_s),
		\end{split}
	\end{align}
	where (a) uses the definition of $J_i$ in \eqref{eq:online-J}, and (b) is due to the concavity of $\widetilde{Q}_s(\boldsymbol{r}_s)$. In inequality \eqref{eq:app-k2},  $\nabla \widetilde{Q}_{J_i}(\boldsymbol{r}^o_{J_i})^{\top} (\boldsymbol{r}_i^o -\boldsymbol{r}_{J_i}^o)  $  satisfies
	\begin{align}\label{eq:app-k3}
		\begin{split}
			&\nabla \widetilde{Q}_{J_i}(\boldsymbol{r}^o_{J_i})^{\top} (\boldsymbol{r}_i^o -\boldsymbol{r}_{J_i}^o)  
			\leq   \lVert \nabla \widetilde{Q}_{J_i}(\boldsymbol{r}_{J_i}^o)\rVert\lVert \boldsymbol{r}_i^o -\boldsymbol{r}_{J_i}^o  \rVert\\
			\leq & ~ \frac{\lVert \nabla  \widetilde{Q}_{J_i}(\boldsymbol{r}_{J_i}^o)\rVert^2 }{2\eta}+ \frac{\eta\lVert \boldsymbol{r}_i^o -\boldsymbol{r}_{J_i}^o\rVert^2}{2} \\
			\leq & ~ \frac{ \overline{Q}^2 }{2\eta}+ \frac{\eta\lVert \boldsymbol{r}_i^o - \boldsymbol{r}_{J_i}^o\rVert^2}{2},
		\end{split}
	\end{align}
	where $\eta$ is a positive constant. 
	
	Substituting \eqref{eq:app-k3} into \eqref{eq:app-k2} and rearranging it, we obtain
	\begin{align}\label{eq:app-k4}
		\begin{split}
			& \widetilde{Q}_s(\boldsymbol{r}^*_s) \!-\! \widetilde{Q}_{s}(\boldsymbol{r}_s^o) 
			\! \overset{\text{(c)}}{\leq}\! ~ \frac{\lVert\boldsymbol{r}_{J_i}^* \!-\! \boldsymbol{r}_{J_i}^o\rVert^2 \!-\! \lVert \boldsymbol{r}_{J_i}^*   \!-\! \boldsymbol{r}_{i}^o \rVert^2  }{2\alpha} \!+\! \frac{\alpha}{2} \overline{Q}^2,
		\end{split}
	\end{align}
	where (c) follows by setting $\eta = 1/\alpha$, i.e., $\eta/2-1/(2\alpha)=0$. In \eqref{eq:app-k4}, $\lVert\boldsymbol{r}_{J_i}^* - \boldsymbol{r}_{J_i}^o\rVert^2 -\lVert \boldsymbol{r}_{J_i}^*   - \boldsymbol{r}_{i}^o \rVert^2  $  satisfies
	\begin{align}\label{eq:app-k5}
		\begin{split}
			&~\lVert\boldsymbol{r}^*_{J_i} -  \boldsymbol{r}_{J_i}^o\rVert^2 - \lVert \boldsymbol{r}_{J_i}^* - \boldsymbol{r}_i^o\rVert^2\\
			= &~ \lVert\boldsymbol{r}^*_{J_i}-  \boldsymbol{r}_{J_i}^o\rVert^2- \lVert\boldsymbol{r}_{J_i}^o - \boldsymbol{r}_{J_i^{\dag}}^*\rVert^2 + \lVert\boldsymbol{r}_{J_i}^o -\boldsymbol{r}_{J_i^{\dag}}^*\rVert^2 \\
			&~-\lVert \boldsymbol{r}_{J_i}^* -\boldsymbol{r}_i^o\rVert^2\\
			= &~ \lVert\boldsymbol{r}_{J_i}^*-\boldsymbol{r}_{J_i^{\dag}}^*\rVert \lVert\boldsymbol{r}^*_{J_i} - 2\boldsymbol{r}_{J_i}^o + \boldsymbol{r}_{J_i^{\dag}}^*\rVert + \lVert\boldsymbol{r}_{J_i}^o - \boldsymbol{r}_{J_i^{\dag}}^*\rVert^2 \\
			&~- \lVert \boldsymbol{r}_{J_i}^* - \boldsymbol{r}_i^o\rVert^2\\
			\overset{\text{(d)}}{\leq} & ~ 2R \lVert\boldsymbol{r}_{J_i}^*- \boldsymbol{r}_{J_i^{\dag}}^*\rVert + \lVert\boldsymbol{r}_{J_i}^o - \boldsymbol{r}_{J_i^{\dag}}^*\rVert^2-\lVert \boldsymbol{r}_{J_i}^* -\boldsymbol{r}_i^o\rVert^2,
		\end{split}
	\end{align}
	where (d) is  due to $\lVert \boldsymbol{r}_i - \boldsymbol{r}_j\rVert\leq R,~\forall \boldsymbol{r}_i,\boldsymbol{r}_j\in{\mathcal{R}}^{K}$.
	
	Substituting \eqref{eq:app-k5} into \eqref{eq:app-k4}, we obtain  inequality  \eqref{eq:app-lemma1}.

	\emph{Proof for Statement (B):} We first introduce an index  $\widehat{I} = \{\widehat{I}>I~|~\mathcal{I}_{\widehat{I}}\neq\emptyset,  \mathcal{I}_{I+1}=\mathcal{I}_{I+2}=\cdots=\mathcal{I}_{\widehat{I}-1}=\emptyset\}$, i.e., the first segment after segment $I$ such that $\mathcal{I}_{\widehat{I}}\neq\emptyset$. 
	Given the enlarged set of segments, i.e., $\mathcal{I}\cup\{I+1,\cdots,\widehat{I}\}$, according to  inequality \eqref{eq:app-lemma1} in Lemma \ref{lem:app1},  we  obtain \eqref{eq:app-lemma2}.
\end{proof}

Based on Lemma \ref{lem:app1}, we show that the dynamic regret of Algorithm \ref{alg:saddle-point} is upper-bounded in Theorem \ref{thm:regret}.
\begin{theorem}[Dynamic Regret]\label{thm:regret}
	Under Condition \ref{ass:parameters} and convex decision set $\mathcal{R}$, the dynamic regret of Algorithm \ref{alg:saddle-point} is upper-bounded by 
	\begin{multline}\label{eq:regret-bound}
		\text{Reg}_I \leq \frac{R^2\left(1+V_{\emptyset}\right)}{2\alpha}   + \frac{RV_{r}}{\alpha} + \frac{\alpha \overline{Q}^2 I}{2}\\+ \mathbbm{1}({\widetilde{\mathcal{I}}}\neq \emptyset)\frac{K R_{|\mathcal{R}|}}{d^{\emph{\text{min}}}} \left( \frac{3R^2}{2\alpha} +\frac{\alpha}{2} \overline{Q}^2\right).
	\end{multline}
\end{theorem}
\begin{proof}
	Based on Lemma \ref{lem:app1}, taking the summation over all the segments $s\in\mathcal{I}_i$ for all $i\in\mathcal{I}\cup\{I+1\}$, we have 
	\begin{align}\label{eq:app-regret1}
		\begin{split}
			& ~\sum_{i\in\mathcal{I}\cup\{I+1\}}\sum_{s\in\mathcal{I}_i}\left(\widetilde{Q}_{s}(\boldsymbol{r}_s^o)  - \widetilde{Q}_s(\boldsymbol{r}^*_s)\right)\\ 
			{\leq} & ~ \sum_{i\in\mathcal{I}\cup\{I+1\}} \frac{1}{2\alpha}|\mathcal{I}_i| \big(2R \lVert\boldsymbol{r}_{J_i}^*- \boldsymbol{r}_{J_i^{\dag}}^*\rVert \\
			&+ \lVert\boldsymbol{r}_{J_i}^o - \boldsymbol{r}_{J_i^{\dag}}^*\rVert^2-\lVert \boldsymbol{r}_{J_i}^* -\boldsymbol{r}_i^o\rVert^2\big)+\frac{\alpha \overline{Q}^2 I}{2} \\
			\overset{\text{(a)}}{\leq} &~\sum_{i\in\mathcal{I}\cup\{I+1\}}\frac{1}{2\alpha}|\mathcal{I}_i| \big(\lVert\boldsymbol{r}_{J_i}^o -  \boldsymbol{r}_{J_i^{\dagger}}^*\rVert^2-\lVert \boldsymbol{r}_{J_i}^* -\boldsymbol{r}_i^o\rVert^2\big) \\
			&~ + \frac{RV_{r}}{\alpha} + \frac{\alpha \overline{Q}^2 I}{2} \\
			\overset{\text{(b)}}{\leq} & ~
			\frac{R^2\left(1+V_{\emptyset}\right)}{2\alpha} + \frac{RV_{r}}{\alpha} + \frac{\alpha \overline{Q}^2 I}{2} .
		\end{split}
	\end{align}
	Inequality (a) is due to the condition in \eqref{eq:assumption-3}. Inequality (b) holds because the following inequality holds:
	\begin{align}\label{eq:app-k8}
		\begin{split}
			&\sum_{i\in\mathcal{I}\cup\{I+1\}} |\mathcal{I}_i|\big(\lVert\boldsymbol{r}_{J_i}^o - \boldsymbol{r}_{J_i^{\dag}}^*\rVert^2-\lVert \boldsymbol{r}_{J_i}^* -\boldsymbol{r}_i^o\rVert^2\big)\\
			\leq & ~\sum_{i\in\mathcal{I}\cup\{I+1\}}\left( 1+ \left[|\mathcal{I}_i| - 1\right]^+ \right)R^2 
			\leq  (1+V_{\emptyset})R^2.
		\end{split}
	\end{align}
	Taking the summation over all the segments $s\in\widetilde{\mathcal{I}}$, we have 
	\begin{align}\label{eq:app-regret2}
		\begin{split}
			&~\sum_{s\in\widetilde{\mathcal{I}}}\left(\widetilde{Q}_s(\boldsymbol{r}^*_s) -\widetilde{Q}_{s}(\boldsymbol{r}_s^o) \right)
			\overset{\text{(c)}}{\leq} 
			\frac{KR_{|\mathcal{R}|}}{d^{\text{min}}} \left( \frac{3R^2}{2\alpha} +\frac{\alpha}{2}\overline{Q}^2\right),
		\end{split}
	\end{align}
	where (c) holds because $(K {R_{|\mathcal{R}|}  \beta})/({ d^{\text{min}} } {\beta}) = {KR_{|\mathcal{R}|} }/{d^{\text{min}} }$ is the upper-bound of the cardinality of $\widetilde{\mathcal{I}}$. By summing up  \eqref{eq:app-regret1} and \eqref{eq:app-regret2}, we can obtain the regret bound \eqref{eq:regret-bound}.
\end{proof}

\subsubsection{Sublinear Dynamic Regret}
Based on Theorem \ref{thm:regret}, the dynamic regret of Algorithm \ref{alg:saddle-point} is sublinear under particular setting  of $\alpha$.
\begin{corollary}[Sublinear Dynamic Regret]\label{coro:regret}
	Under Condition \ref{ass:parameters} and convex decision set $\mathcal{R}$, by  setting $\alpha = \alpha_0I^{-{1}/{\gamma}}$, for any $\gamma\in(1,\infty)$, the dynamic regret of Algorithm \ref{alg:saddle-point} is sublinear, i.e., $\lim_{I\rightarrow\infty} \text{Reg}_{I} / I = 0$.
\end{corollary}
\begin{proof}
	For any $\gamma\in(1,\infty)$, 
	\begin{multline}
		\lim_{I\rightarrow\infty}\text{Reg}_I/I \leq \lim_{I\rightarrow\infty} \frac{R^2\left(1\!+\!V_{\emptyset}\right)}{2\alpha_0I^{-\frac{1}{\gamma}}I } + \frac{RV_{r}}{\alpha_0I^{-\frac{1}{\gamma}} I} +\frac{\alpha_0I^{-\frac{1}{\gamma}} \overline{Q}^2}{2} \\+ \mathbbm{1}({\widetilde{\mathcal{I}}}\neq \emptyset)\frac{K R_{|\mathcal{R}|}}{d^{\text{min}}} \left( \frac{3R^2}{2\alpha_0I^{-\frac{1}{\gamma}} I} +\frac{\alpha_0I^{-\frac{1}{\gamma}} }{2I} \overline{Q}^2\right) = 0,
	\end{multline}
	and $\lim_{I\rightarrow\infty}\text{Reg}_I/I \geq 0$.
\end{proof}

Corollary \ref{coro:regret} implies that as the number of segments increases, the performance of the proposed online algorithm approaches  the performance under the optimal solutions to the per-segment problems (OPT-SEGMENT-i) for $i\in\mathcal{I}$.

\subsection{Algorithm Modification}


Algorithm \ref{alg:saddle-point} is capable of learning the user's FoV preference and time-varying downloading capacity in real time, and it is proven to have sublinear dynamic regret under Condition \ref{coro:regret}. In practice, however, the downloading capacity may sometimes have significant increase or decrease within a short time (with statistics from a real-world dataset to be shown in Section \ref{subsec:data}), under which Condition  \ref{coro:regret} may sometimes not be satisfied, and this may induce sudden changes to the bitrate decisions.

To alleviate the impact of the variation of the downloading capacity, we modify Algorithm \ref{alg:saddle-point} by restricting the increase or decrease of the bitrate of each tile to be at most one level at each time. Specifically, after computing $\boldsymbol{r}^o_i $ for segment $i\in\mathcal{I}$ according to \eqref{eq:online-update}, we compute  a modified bitrate decision $\tilde{\boldsymbol{r}}^o_i = (\tilde{{r}}^o_{i,k}, k\in\mathcal{K})$, and use it as the bitrate choice of the online algorithm, i.e., assign the value of vector $\tilde{\boldsymbol{r}}^o_i $ to vector ${\boldsymbol{r}}^o_i  $. The bitrate vector $\tilde{\boldsymbol{r}}^o_i$ is defined as follows. Let $\tilde{l}^o_{i,k}$ and ${l}^o_{i,k}$ denote the bitrate level of the bitrate vectors $\tilde{r}^o_{i,k}$ and ${r}^o_{i,k}$, i.e., $R_{\tilde{l}^o_{i,k}} = \tilde{r}^o_{i,k}$ and $R_{{l}^o_{i,k}} = {r}^o_{i,k}$, respectively. Given bitrate vectors ${\boldsymbol{r}}^o_i$ and $\tilde{\boldsymbol{r}}^o_{i-1}$, vector $\tilde{\boldsymbol{r}}^o_i$  is computed as follows: for all $i\in\mathcal{I}$ and $k\in\mathcal{K}$,
\begin{equation}\label{eq:modify}
	\tilde{r}^o_{i,k} = \left\{
	\begin{array}{ll}
		R_{\tilde{l}^o_{i-1,k} + 1},& {l}^o_{i,k} > \tilde{l}^o_{i-1,k} + 1, \\
		R_{\tilde{l}^o_{i-1,k} - 1},&{l}^o_{i,k} < \tilde{l}^o_{i-1,k} - 1,\\
		{r}^o_{i,k},& \text{otherwise}.\\
	\end{array}
	\right.
\end{equation}
Intuitively, when the bitrate decision of a tile computed according to \eqref{eq:online-update} is increased (or decreased) by more than one level when compared with the bitrate of the tile of the previous segment, we restrict the increase (or decrease) of the bitrate to be one level so as to avoid the aggressive bitrate change caused by the significant variation of the downloading capacity.


\section{Performance Evaluation}\label{sec:simulation}
In this section, we evaluate the performance of our proposed OBS360 algorithm. In the simulations, we use two open datasets to simulate the 360 degree video streaming scenarios: we use the dataset in \cite{van2016http} to simulate the users' downloading capacities, and use the dataset in \cite{Knorr2018data} to simulate users' FoVs. 
The coefficients are as follows: $l^{\textsc{rb}} =0.5$, $l^{\textsc{bd-e}}=0.1$, and $l^{\textsc{bd-a}}=0.1$. In addition,  initial buffer occupancy $b^{\textsc{ini}}$ is set to two seconds,  segment length $\beta$ is set to one second, and parameter $\alpha$ is set to one. 

In the following, we first introduce the two open datasets from \cite{van2016http} and \cite{Knorr2018data}. Then, we present the simulation results, including the video streaming instance using the proposed OBS360 algorithm, the comparison between OBS360 algorithm and the offline optimal performance, and the comparison between OBS360 algorithm and benchmark methods.

\subsection{Datasets}\label{subsec:data}

The dataset in \cite{van2016http} contains the bandwidth measurement results in 4G networks  in the city of Ghent, Belgium. There are a total of 40 logs, which are collected on various transportation, such as on foot, bicycle, bus, and train. Each log has a duration ranging from 166 to 758 seconds. In our simulation, we select one sample within each second to compute the downloading capacity of the second.  Fig. \ref{fig:down-data} (a) shows the cumulative distribution function (CDF) of the downloading capacity of these selected samples. As shown in the figure, $50\%$ of the downloading capacities are below 33 Mbps, and $80\%$ of them are below 48 Mbps. Fig. \ref{fig:down-data} (b) shows the CDF of the absolute downloading capacity difference between adjacent selected samples. As shown in the figure, more than $20\%$ of the adjacent samples have a downloading capacity difference that is larger than $15$ Mbps. This implies that the downloading capacity can have significant sudden changes, which makes it challenging for a video player to decide the bitrate of each tile for VA360 video streaming services.
\begin{figure}[t]
	\centering
	\includegraphics[height=3.3cm]{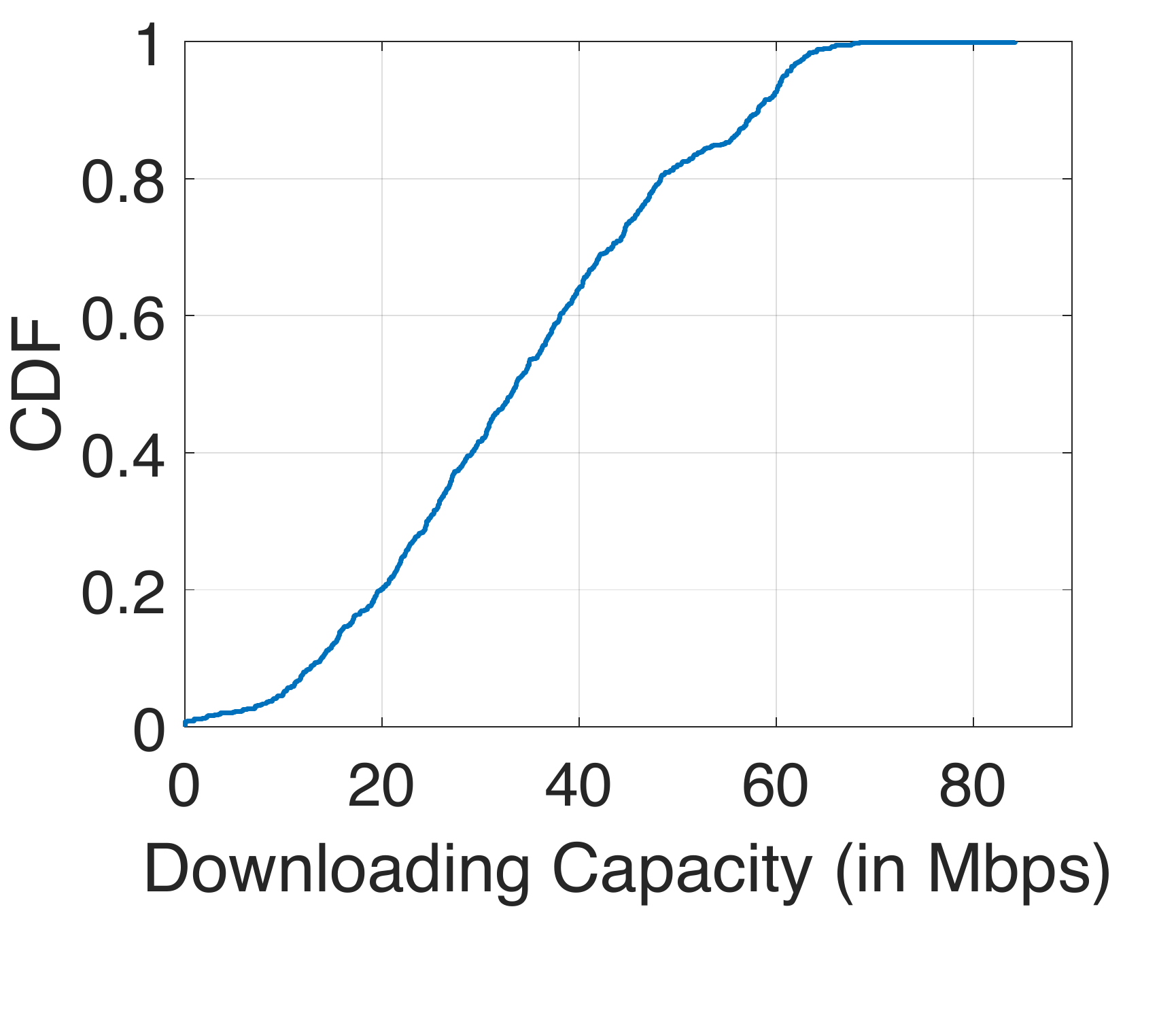}
	\includegraphics[height=3.3cm]{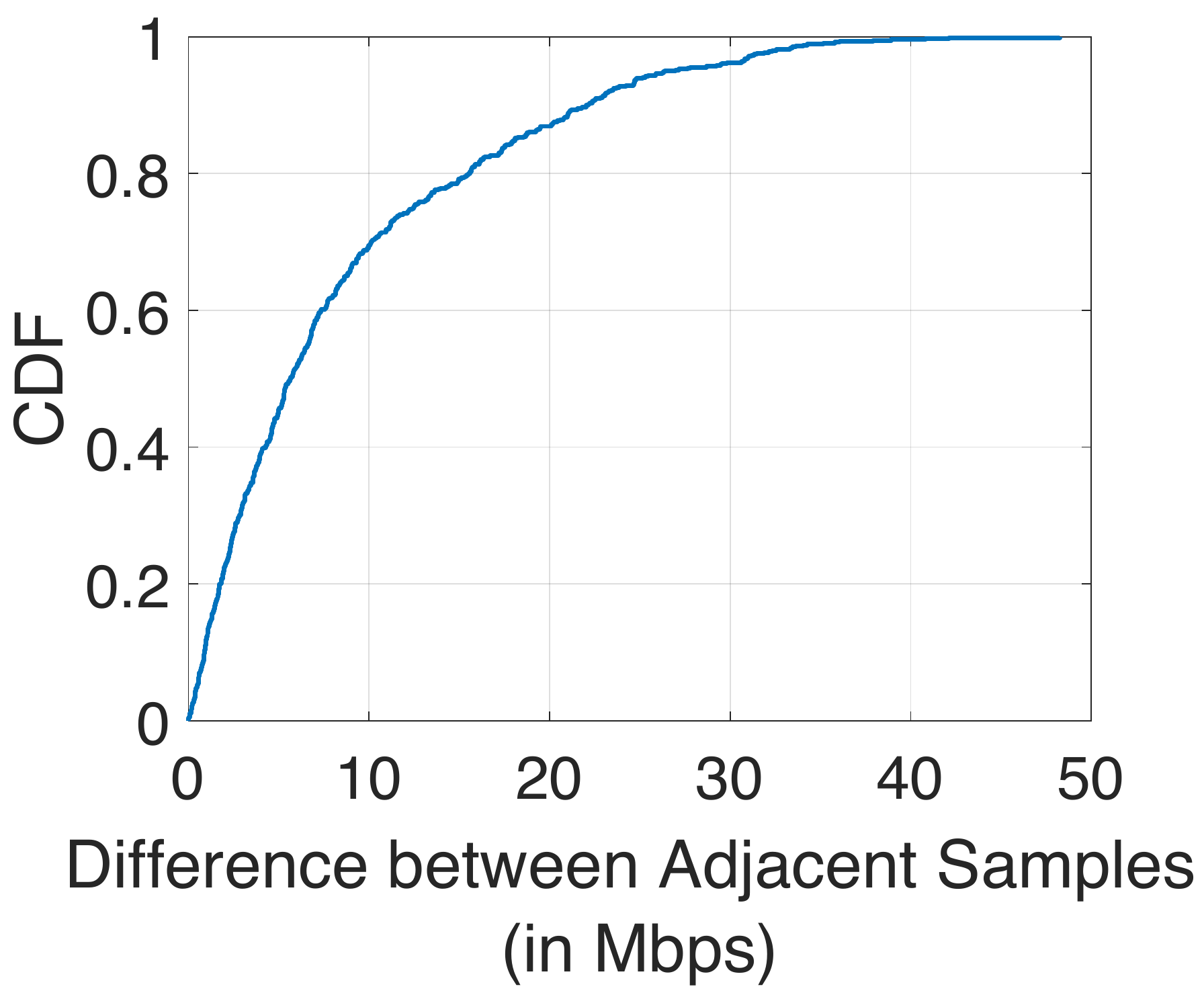}\\
	(a)\qquad\qquad\qquad\qquad\qquad(b)
	\caption{Dataset statistics: (a) downloading capacity; (b) absolute downloading capacity difference between adjacent samples.}\label{fig:down-data}
\end{figure}
\begin{figure}[t]
	\centering
	\includegraphics[height=2.2cm]{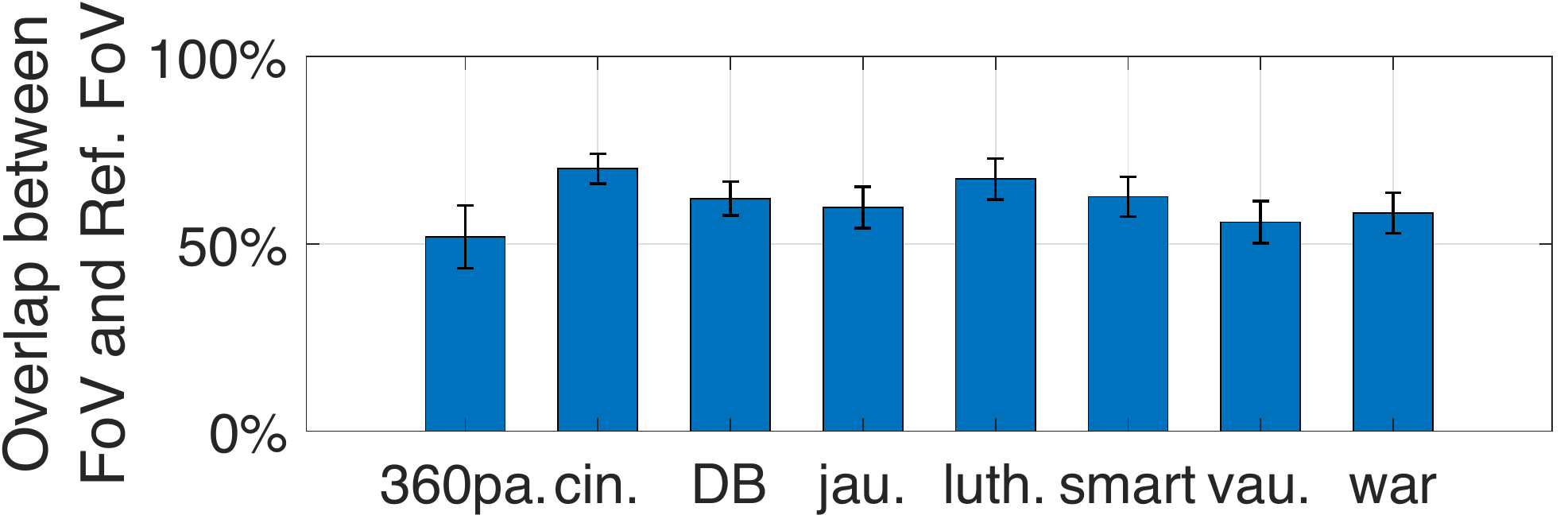}
	\caption{Overlap percentage between the users' FoVs and the reference FoV. The bars and the error bars give the average values and the variance across users, respectively.}\label{fig:video-data}
\end{figure}

The dataset in \cite{Knorr2018data} contains the traces of 8 videos. For each video, there is one recommended viewing trace (marked by professional filemakers) and 20 actual viewing  traces from 20 users. 
In the dataset, either the recommended trace or each user's trace of watching a video is represented by a viewing degree trace $((p_1,y_1),(p_2,y_2),\ldots,(p_I,y_I))$, where $p_i\in[-90^{\circ},90^{\circ}]$ and $y_i\in[-180^{\circ},180^{\circ}]$ are the pitch (vertical degree)  and yaw (horizontal degree) of the corresponding viewport of segment $i=1,2,\ldots,I$, respectively. The actual viewing degree trace is then transferred to a viewing FoV trace $\boldsymbol{\omega}=(\omega_{i,k},i\in\mathcal{I},k\in\mathcal{K})$ according to the definition of FoV in Section \ref{subsec:qoe}, with the recommended viewing FoV as the reference FoV. Fig. \ref{fig:video-data} shows the average overlap between the users' FoVs and the reference FoV of different videos, where the x-axis corresponds to the titles or title abbreviations of the videos. As shown in Fig. \ref{fig:video-data}, these videos have different average overlap percentages, while most of the percentages are around $60\%$. 

\subsection{Streaming Instance Using OBS360 Algorithm}\label{subsec:instance}
\begin{figure}
	\centering
	~\includegraphics[height=2.8cm]{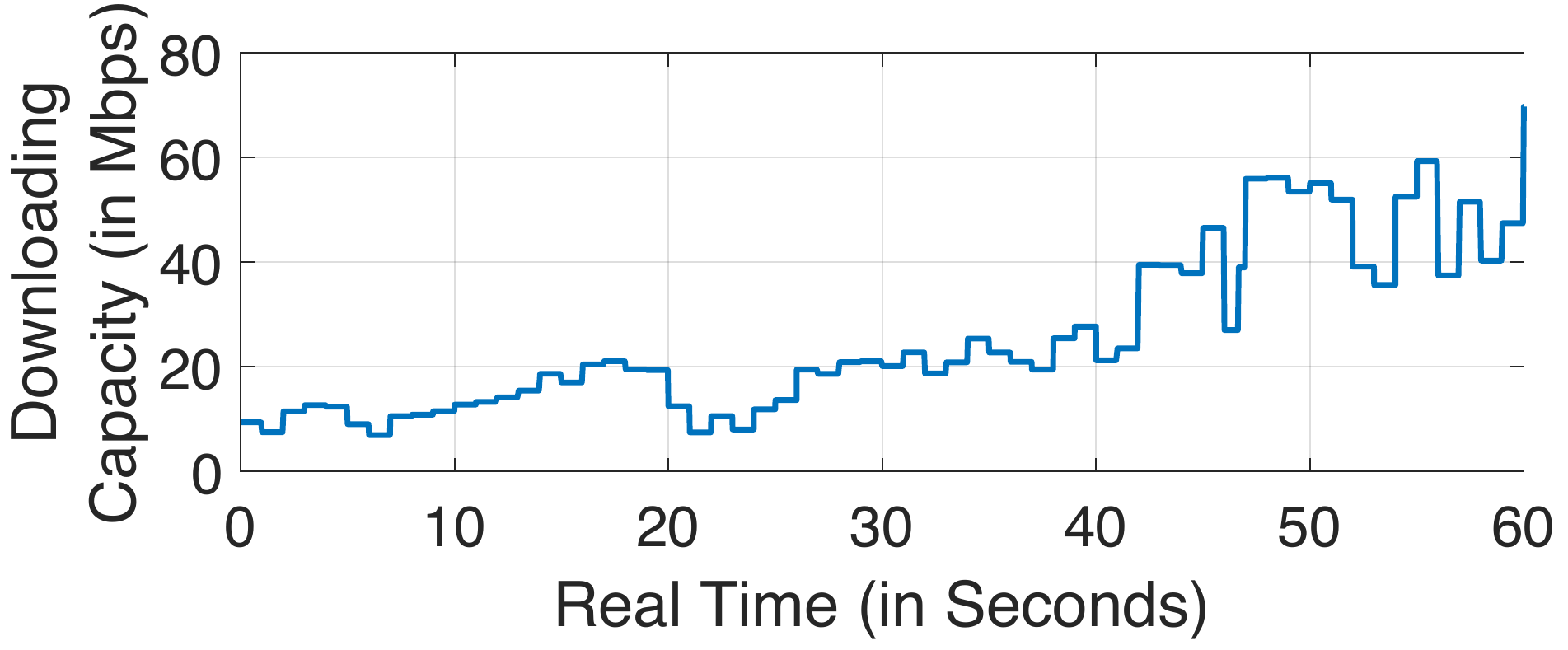}\\
	~~~~(a)\\
	\includegraphics[height=2.8cm]{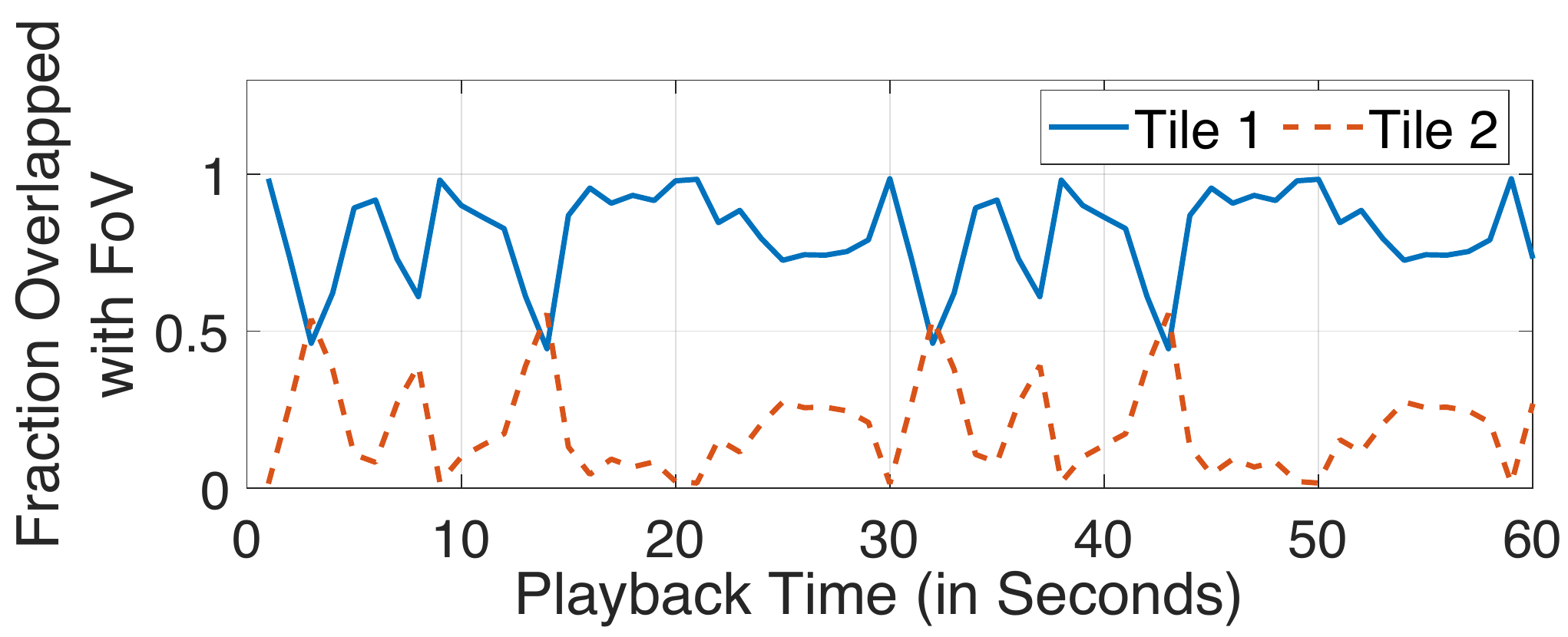}\\
	~~~~(b)\\
	~~~\includegraphics[height=2.7cm]{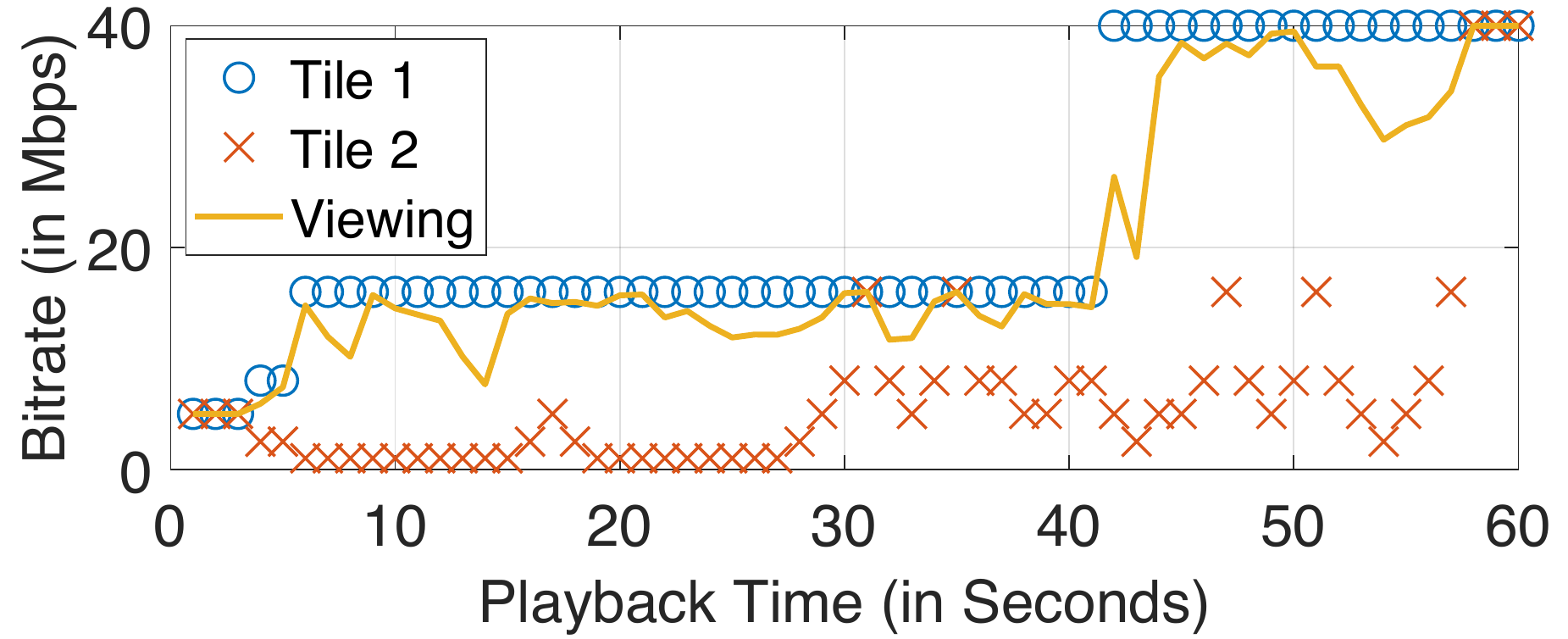}\\
	~~~~(c)\vspace{1mm}\\
	\includegraphics[height=2.6cm]{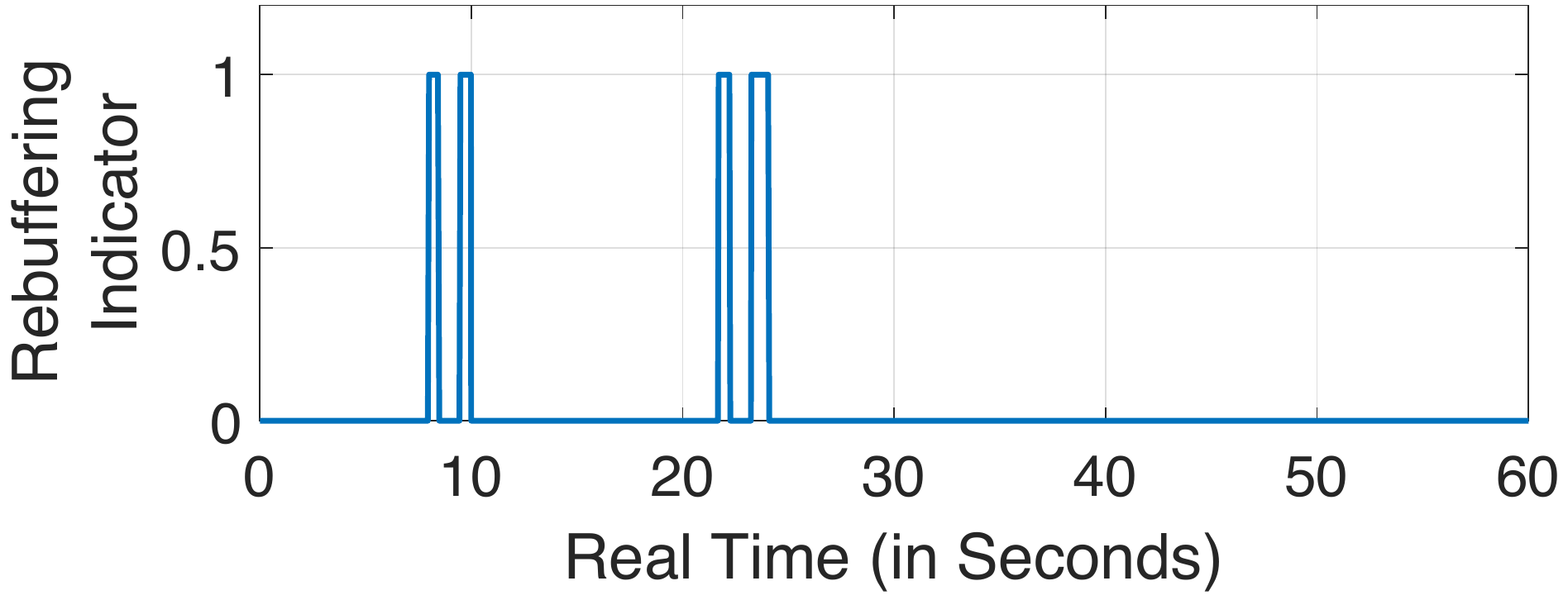}~~~\\
	~~~~(d)\\
	\caption{A video streaming instance with OBS360 algorithm: (a) downloading capacity; (b) the fraction of the tile overlapped with FoV; (c) tile bitrate and viewing bitrate; (d) rebuffering indicator.}\label{fig:instance}
\end{figure}
To visualize the bitrate selection of the proposed OBS360 algorithm, in Fig. \ref{fig:instance}, we show a video streaming scheduling instance under a downloading capacity trace from the dataset in \cite{van2016http} and a FoV trace from the dataset in \cite{Knorr2018data}. In this simulation, we consider a simplified setting of two tiles placed in a one by two grid. This simplification is for illustrating the bitrate selection of each tile and obtaining intuitions from the video streaming scheduling. We consider each user can see half of the whole view at any time, and the available bitrate set for each tile is  $\mathcal{R}=\{1, 2.5, 5, 8, 16, 40\}$ (Mbps). As a result, if all the tiles for a segment have the same bitrate from set $\mathcal{R}$, the viewing bitrate of the segment is in set $\{1, 2.5, 5, 8, 16, 40\}$ (Mbps), as the recommended bitrate set in YouTube \cite{youtube-dash}. The initial bitrate of each tile is set to be 5 Mbps. Note that such a setting is for demonstration, and a more realistic case with 16 tiles, as in \cite{Xiao2018-BAS360,Qian2018-Flare}, will be evaluated in Section \ref{subsec:comparison}.

Fig. \ref{fig:instance} (a) shows a downloading capacity trace, where the difference between adjacent capacity samples can be up to 20 Mbps. This trace is from \cite{van2016http} and is modified to have an increasing trend across the real time for the convenience of performance demonstration. Fig. \ref{fig:instance} (b) shows the user's FoV trace. This FoV varies dramatically across the playback time, and majority parts of the FoV are on tile 1 at most of the playback time. Fig. \ref{fig:instance} (c) shows the bitrate selection result of the proposed OBS360 algorithm. As shown in the figure, the bitrates of both tiles are initialized as 5 Mbps, so the viewing bitrate is 5 Mbps at the beginning of the streaming. Through learning, the algorithm gradually increases the bitrate of tile 1 from 5 Mbps to 16 Mbps at around 6 seconds (playback time), and hence the user's viewing bitrate significantly increases to around 16 Mbps (approximately a 2K video \cite{youtube-dash}). After around 42 seconds (real time), the downloading capacity significantly increases to around 40 Mbps, so the bitrate of tile 1 further increases, and the user's viewing bitrate increases to around 40 Mbps (approximately a 4K video \cite{youtube-dash}). In terms of the rebufferring, it happens at around 8, 9, 22, and 23 seconds (real time), which is mainly due to the sudden decrease of the downloading capacity at around 7 and 21 seconds. Note that despite the significant varying of the downloading capacity in the trace, each of the rebuffering lasts only a few hundred milliseconds, which is hardly noticeable by the user \cite{de2012quantifying}.

\subsection{OBS360 and Offline Optimal Solution}\label{subsec:comparison-offf}
\begin{figure}[t]
	\centering
	~~\includegraphics[height=2.4cm]{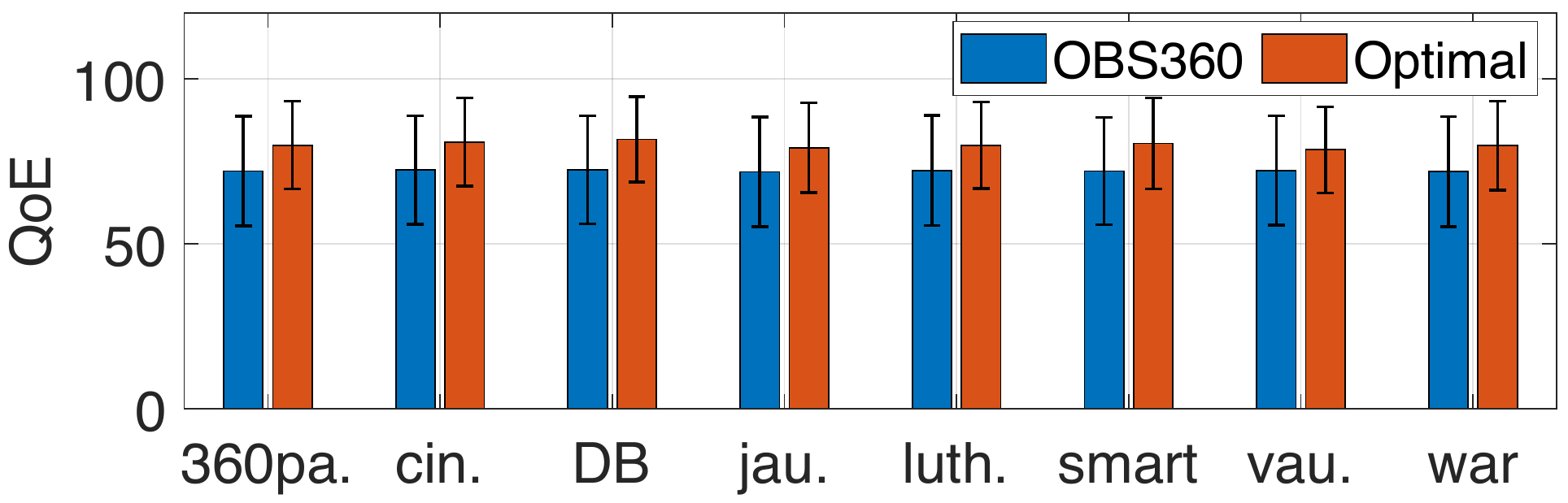}\\
	~~~~(a)\\
	\includegraphics[height=2.55cm]{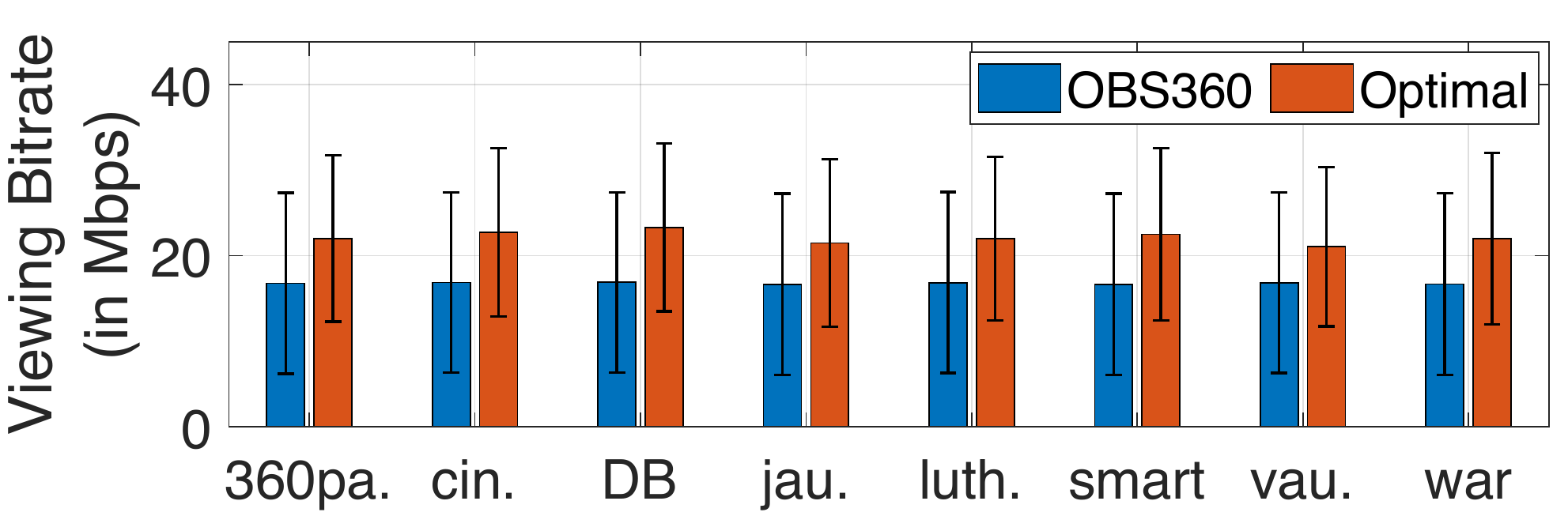}\\
	~~~~(b)\\
	\includegraphics[height=2.7cm]{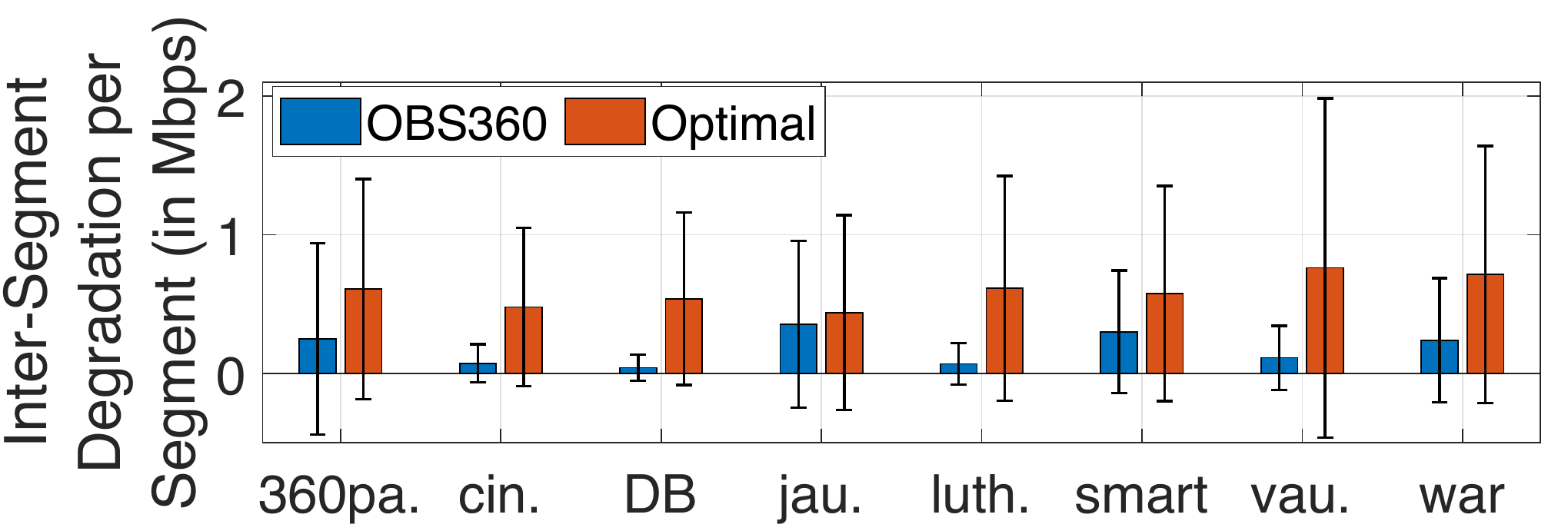}~~~\\
	~~~~(c)\\
	\includegraphics[height=2.65cm]{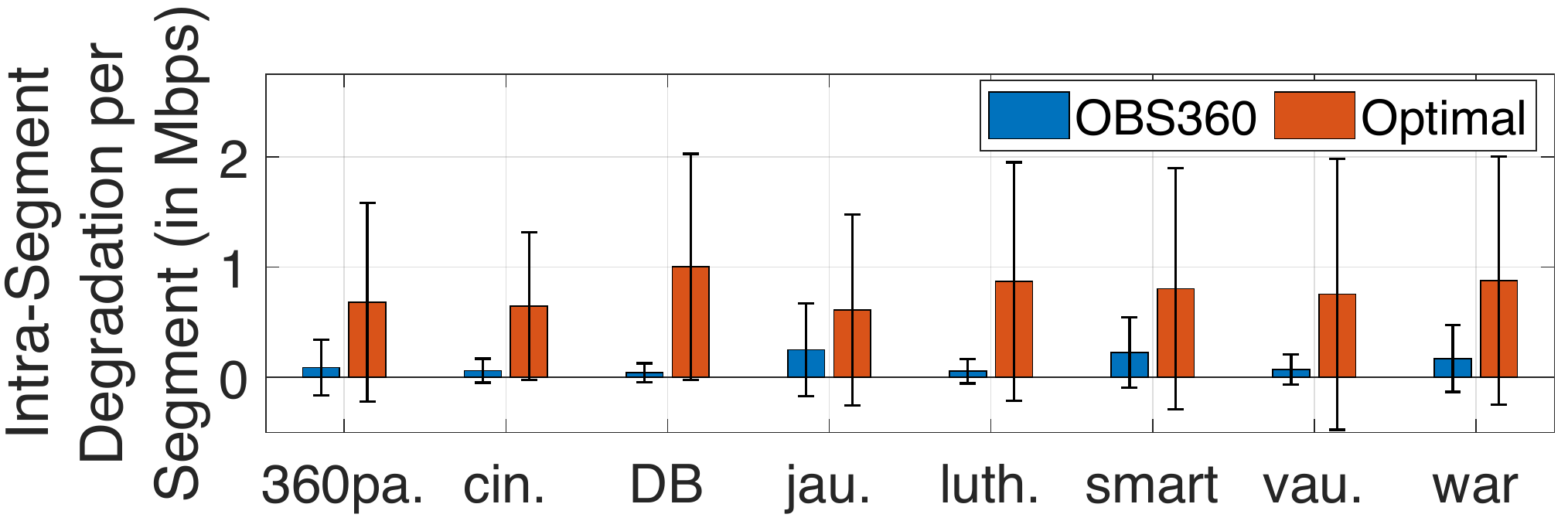}~~~\\
	~~~~(d)\\
	\includegraphics[height=2.55cm]{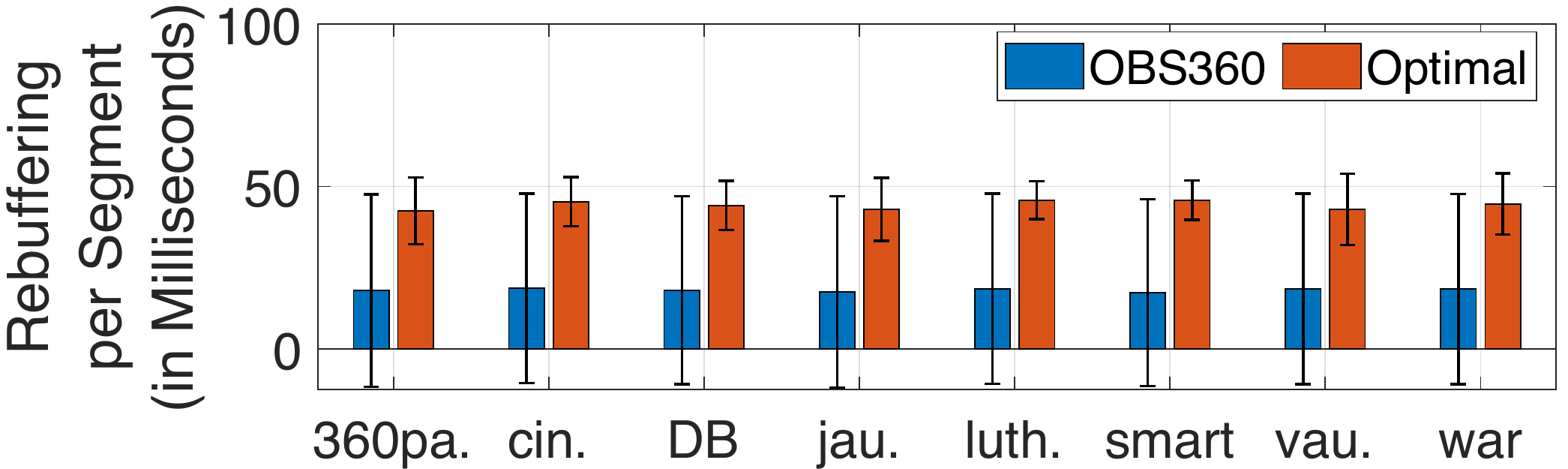}~~~~~~~~~~\\
	~~~~(e)\\
	\caption{{Comparisons between OBS360 and offline optimal solution: (a) QoE; (b)  viewing bitrate; (c) inter-segment degradation; (d)  intra-segment degradation; (e) rebuffering.}}\label{fig:comp-off}
\end{figure}

We compare our proposed OBS360 algorithm with the offline optimal solution, i.e., the optimal solution to problem \eqref{eq:optimization} when the user's FoV and downloading capacity are known beforehand. 
Due to the computational complexity of computing the offline optimal solution, we consider the simplified VA360 video streaming setting as in Section \ref{subsec:instance}. The simulations are performed based on the traces of all eight videos in the dataset from \cite{Knorr2018data}. For each video, we consider 20 users, each corresponding to a downloading capacity trace from the dataset in \cite{van2016http}  and a FoV trace from the dataset in \cite{Knorr2018data}. The performance is shown in Fig. \ref{fig:comp-off}. 
The bars show the corresponding average values over the users with different FoVs and downloading capacities, and the error bars show the corresponding variance across the users.

As shown in Fig. \ref{fig:comp-off} (a), the proposed OBS360 algorithm achieves $90.1\%$ of the user's QoE of the offline optimal solution on average. Specifically, the proposed algorithm achieves a viewing bitrate of around 16 Mbps (approximately a 2K video \cite{youtube-dash}), which is $75.8\%$ of the viewing bitrate of the offline optimal performance, as shown in Fig. \ref{fig:comp-off} (b). In addition, it achieves a lower inter-segment degradation, a lower intra-segment degradation, and a lower rebuffering than the offline optimal solution, as shown in Figs. \ref{fig:comp-off} (c)-(e). Intuitively, when compared with the offline optimal performance, although the proposed OBS360 algorithm has a lower viewing bitrate, it avoids frequent rebuffering and bitrate degradation.

\subsection{OBS360 and Benchmark Methods}\label{subsec:comparison}

\begin{figure}[t]
	\centering
	\includegraphics[height=2.35cm]{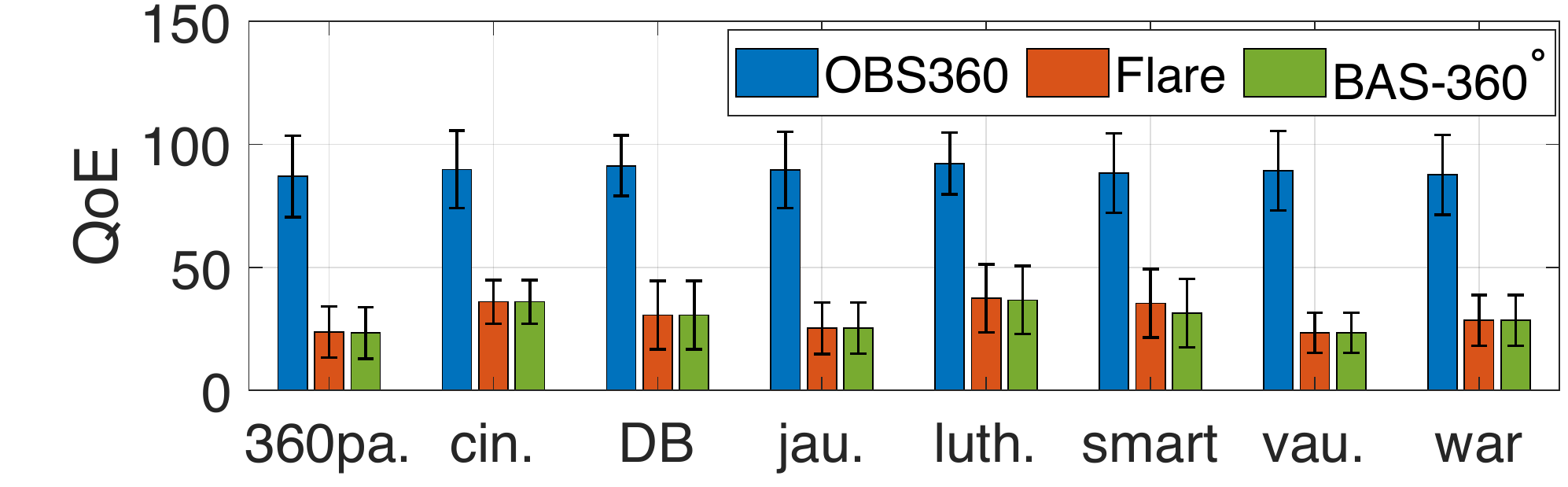}\\
	~~~~(a)\\
	~\includegraphics[height=2.4cm]{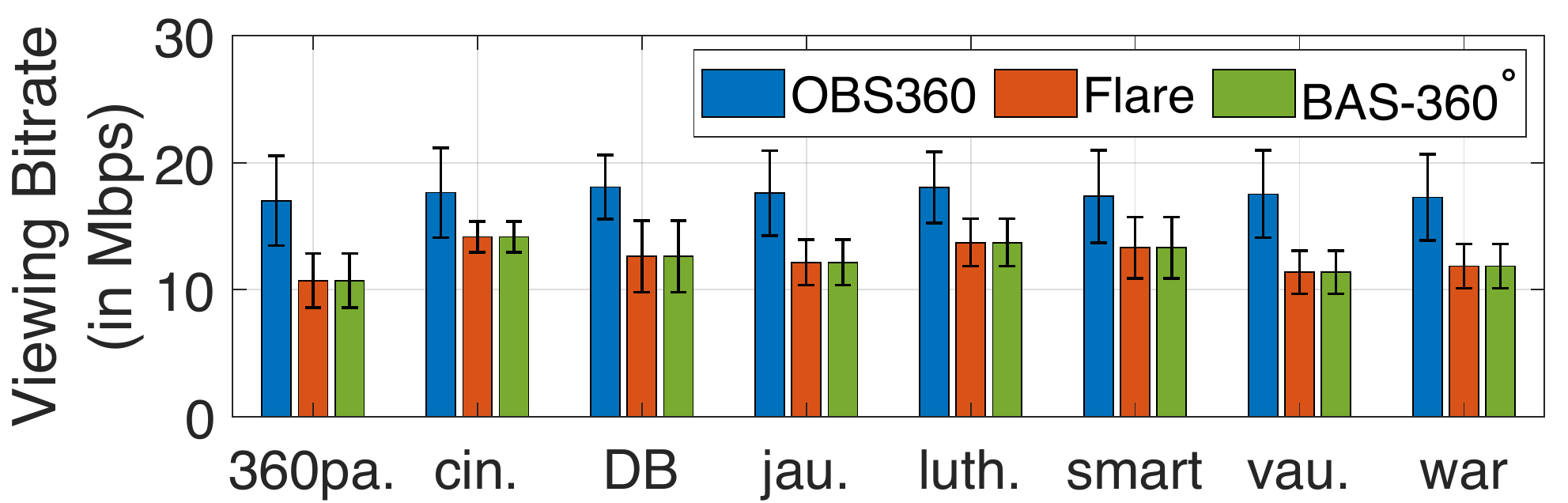}\\
	~~~~(b)\\
	\includegraphics[height=2.55cm]{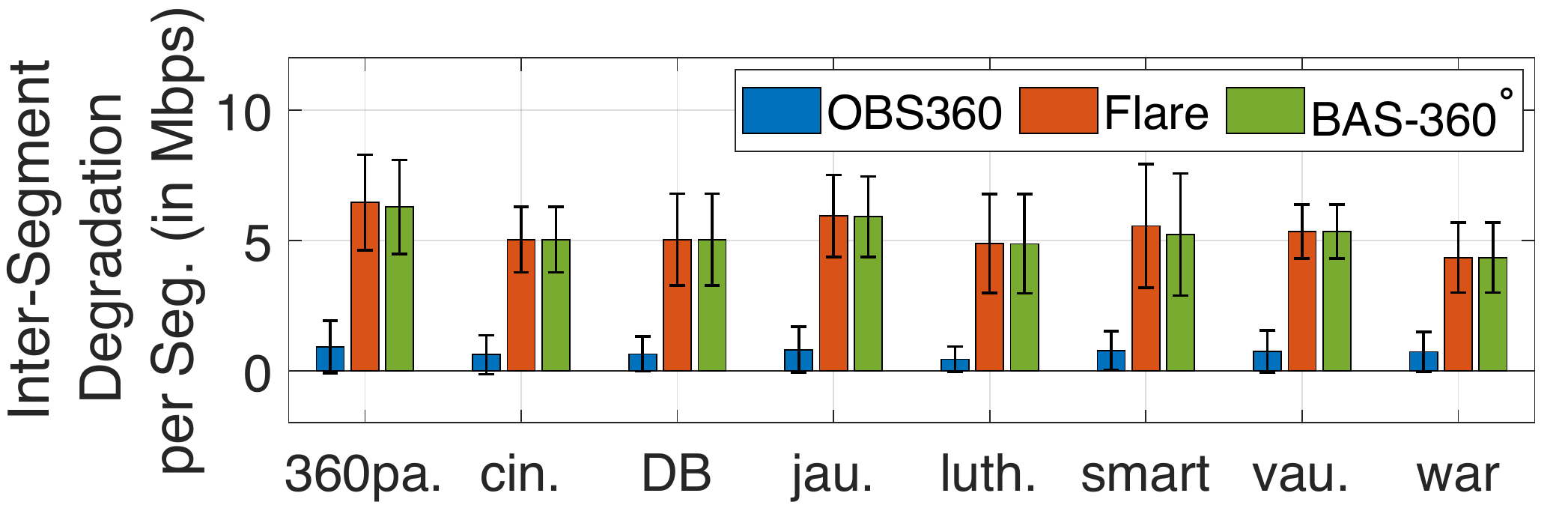}~~\\
	~~~~(c)\\
	\includegraphics[height=2.6cm]{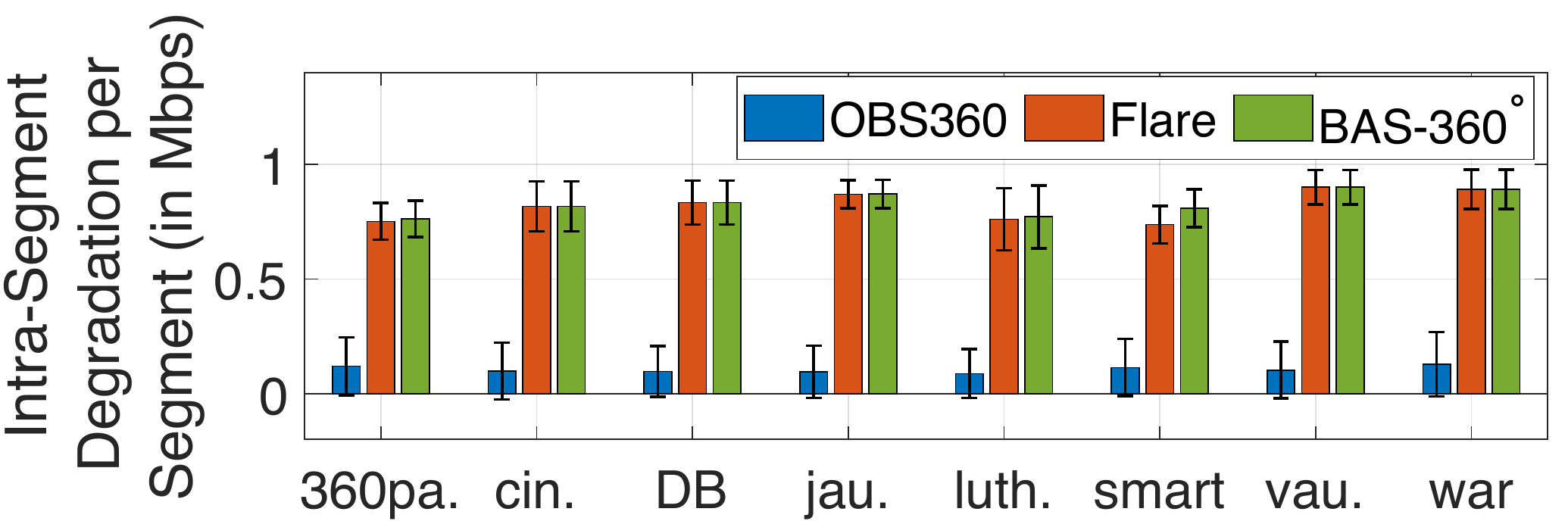}~~~\\
	~~~~(d)\\
	\includegraphics[height=2.45cm]{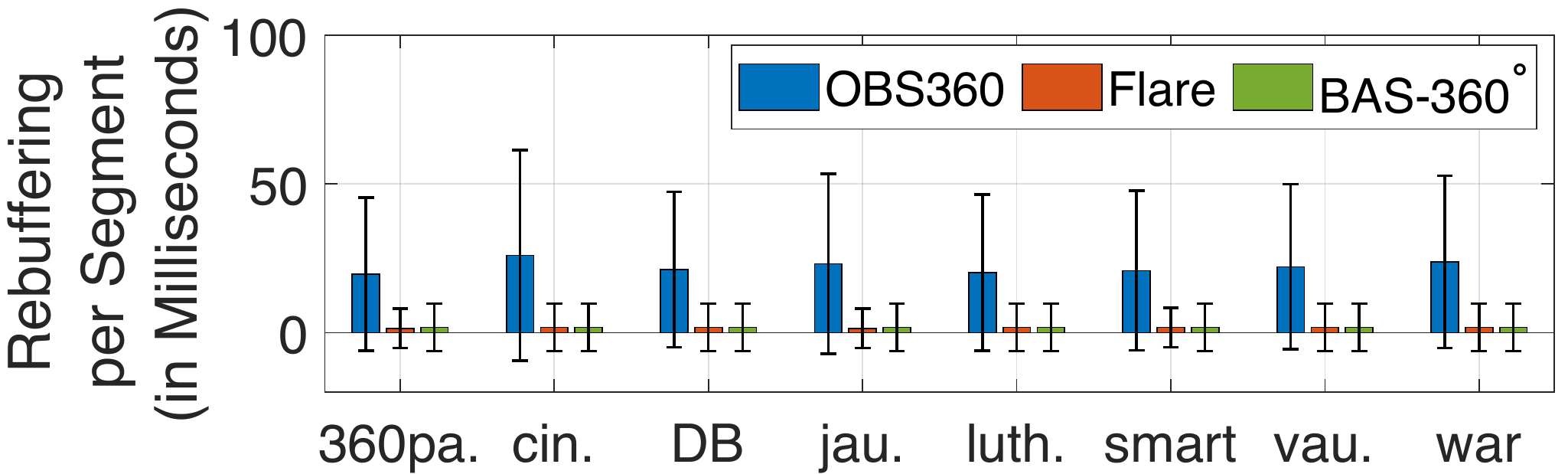}~~~~\\
	~~~~(e)\\
	\caption{{Comparisons between OBS360 and benchmark methods: (a) QoE; (b)  viewing bitrate; (c) inter-segment degradation; (d)  intra-segment degradation; (e) rebuffering.}}\label{fig:comp}
\end{figure}

We perform simulations with the open datasets from  \cite{van2016http} and \cite{Knorr2018data} to compare the performance of our proposed OBS360 algorithm with  the BAS-360$^{\circ}$ algorithm in \cite{Xiao2018-BAS360} and the Flare algorithm in \cite{Qian2018-Flare}. The Flare and BAS-360$^{\circ}$  algorithms are bitrate selection algorithms based on user viewing history. The Flare algorithm aims to maximize the predicted viewing bitrate subject to bandwidth constraints, and BAS-360$^{\circ}$  algorithm aims to minimize the bandwidth waste (i.e., the bandwidth that can be further utilized without inducing rebuffering and the bandwidth used for downloading the unviewed tiles).





In our simulations,  each segment is divided into 16 tiles in a four by four grid, and a user can see one quarter of the whole view at any time, as in \cite{Xiao2018-BAS360,Qian2018-Flare}. The available bitrate set for each tile is  $\mathcal{R} = \{0.25, 0.625, 1.25, 2, 4, 10\}$ (Mbps). In this case, if all the tiles for a segment have the same bitrate from set $\mathcal{R}$, the viewing bitrate of the segment is within set  $\{1, 2.5, 5, 8, 16, 40\}$ (Mbps), which is consistent with the recommended bitrate set in YouTube\cite{youtube-dash}. We perform the simulations based on the traces of all eight videos \cite{Knorr2018data}, and consider 20 users, each corresponding to a downloading capacity trace from the dataset in \cite{van2016http} and a FoV trace from the dataset in \cite{Knorr2018data}. Fig. \ref{fig:comp} shows the performance comparisons between our proposed algorithm and the benchmark methods. 

Fig. \ref{fig:comp} (a) shows the comparison regarding the users' QoE. Specifically, the proposed OBS360 outperforms Flare and BAS-360$^{\circ}$ significantly  for all the videos. In addition, the performances of the benchmark methods highly depend on the average overlap percentages between the users' FoVs and the reference FoV of the videos, while the performance of the proposed algorithm does not. For example, for videos `cin.' and `luth', their overlap percentages are the largest among all videos (in Fig. \ref{fig:video-data}), and their corresponding QoE using the benchmark methods are the largest among all videos (in Fig. \ref{fig:comp} (a)) as well. In comparison, our proposed OBS360 algorithm achieves similar user's QoE under the various videos, regardless of the average overlap percentages of those videos. This shows that the proposed algorithm can learn the users' preferences over the FoV, and adapt the bitrates according to learned preferences.

Fig. \ref{fig:comp} (b) shows that the proposed OBS360 algorithm can improve the viewing bitrate significantly by $24.6\%-58.8\%$ when compared with the benchmark methods. This improvement is mainly due to the capability of the proposed algorithm on learning the users' FoV preferences. Figs. \ref{fig:comp} (c) and (d) show that the proposed OBS algorithm can significantly reduce the inter-segment and intra-segment degradation when compared with the benchmark methods. The inter-segment degradation is reduced by $83.3\%-91.0\%$, and the intra-segment degradation is reduced by $84.1\%-89.1\%$. Fig. \ref{fig:comp} (e) shows the comparison regarding the rebuffering per segment. Although the proposed OBS360 algorithm has larger rebuferring than the Flare and BAS-360$^{\circ}$ algorithms, the absolute values of the rebuffering are quite small (i.e., less than 20 milliseconds per one-second segment), which can hardly be noticed by the users. 

In summary, our proposed algorithm can significantly improve the users' QoE when compared with the benchmark methods. This is achieved  by increasing the user's viewing bitrate and reducing the inter-segment and intra-segment degradation losses.
\section{Conclusion}\label{sec:conclude}
In this work, we considered a scenario with a newly generated 360-degree video without viewing history from other users. We proposed an  OBS360 algorithm to optimize the user's QoE. The proposed online algorithm can adapt to the unknown and heterogeneous users' FoVs and downloading capacities. In addition, the algorithm was proven to have  sublinear dynamic regret under a convex decision set, which provides an intuition that as the number of segments increases, the performance of the online algorithm approaches the offline optimal performance. 
We performed simulations with  real-world datasets regarding the users' FoVs and downloading capacities. The results show that our proposed algorithm achieves $90.1\%$ of the users' QoE of the offline optimal performance. In addition, when compared with Flare and BAS-360$^{\circ}$ algorithms from existing works, our proposed algorithm achieves a higher QoE through improving the viewing bitrate and reducing the inter-segment and intra-segment degradation losses of the users.

The results in this paper can be extended in the following directions. First, the proposed algorithm achieves sublinear dynamics regret under a convex decision set. It is interesting to design an algorithm that can achieve sublinear dynamic regret under a nonconvex decision set. This will be challenging, because the convex decision set is an important condition in online convex optimization for ensuring the sublinearity. Second, it is interesting to design a bitrate selection algorithm by both learning the user's features (i.e., FoV preference and downloading capacity) and exploiting the viewing history of other users.


\bibliographystyle{IEEEtran}
\bibliography{surveyset}

\begin{thebibliography}{10}
\providecommand{\url}[1]{#1}
\csname url@samestyle\endcsname
\providecommand{\newblock}{\relax}
\providecommand{\bibinfo}[2]{#2}
\providecommand{\BIBentrySTDinterwordspacing}{\spaceskip=0pt\relax}
\providecommand{\BIBentryALTinterwordstretchfactor}{4}
\providecommand{\BIBentryALTinterwordspacing}{\spaceskip=\fontdimen2\font plus
\BIBentryALTinterwordstretchfactor\fontdimen3\font minus
  \fontdimen4\font\relax}
\providecommand{\BIBforeignlanguage}[2]{{%
\expandafter\ifx\csname l@#1\endcsname\relax
\typeout{** WARNING: IEEEtran.bst: No hyphenation pattern has been}%
\typeout{** loaded for the language `#1'. Using the pattern for}%
\typeout{** the default language instead.}%
\else
\language=\csname l@#1\endcsname
\fi
#2}}
\providecommand{\BIBdecl}{\relax}
\BIBdecl

\bibitem{Sony}
``{Sony PlayStation VR},''
  \url{https://www.playstation.com/en-au/explore/playstation-vr/}, {Accessed
  on} Feb. 10, 2020.

\bibitem{Hosseini2016}
M.~Hosseini and V.~Swaminathan, ``Adaptive 360 {VR} video streaming: Divide and
  conquer,'' in \emph{Proc. IEEE International Symposium on Multimedia (ISM)},
  San Jose, CA, Dec. 2016.

\bibitem{DAcunto2016}
L.~D'Acunto, J.~van~den Berg, E.~Thomas, and O.~Niamut, ``Using {MPEG DASH SRD}
  for zoomable and navigable video,'' in \emph{Proc. ACM Int'l Conf. on
  Multimedia Systems (MMSys)}, Klagenfurt, Austria, May 2016.

\bibitem{ballard2019rats}
T.~Ballard, C.~Griwodz, R.~Steinmetz, and A.~Rizk, ``{RATS}: {A}daptive
  360-degree live streaming,'' in \emph{Proc. ACM Int'l Conf. on Multimedia
  Systems (MMSys)}, Amherst, MA, Jun. 2019.

\bibitem{Xiao2018-BAS360}
M.~Xiao, C.~Zhou, V.~Swaminathan, Y.~Liu, and S.~Chen, ``{BAS-360${}^{\circ}$}:
  Exploring spatial and temporal adaptability in 360-degree videos over
  {HTTP/2},'' in \emph{Proc. IEEE INFOCOM}, Honolulu, HI, Apr. 2018.

\bibitem{Qian2018-Flare}
F.~Qian, B.~Han, Q.~Xiao, and V.~Gopalakrishnan, ``Flare: Practical
  viewport-adaptive 360-degree video streaming for mobile devices,'' in
  \emph{Proc. ACM MobiCom}, New Delhi, India, Oct. 2018.

\bibitem{zhou2018clustile}
C.~Zhou, M.~Xiao, and Y.~Liu, ``Clustile: Toward minimizing bandwidth in
  360-degree video streaming,'' in \emph{Proc. IEEE INFOCOM}, Honolulu, HI,
  Apr. 2018.

\bibitem{jiang2019hierarchical}
Z.~Jiang, X.~Zhang, W.~Huang, H.~Chen, Y.~Xu, J.~N. Hwang, Z.~Ma, and J.~Sun,
  ``A hierarchical buffer management approach to rate adaptation for 360-degree
  video streaming,'' \emph{IEEE Trans. Veh. Technol.}, 2019 (Early Access).

\bibitem{xie2018cls}
L.~Xie, X.~Zhang, and Z.~Guo, ``{CLS}: A cross-user learning based system for
  improving {QoE} in 360-degree video adaptive streaming,'' in \emph{Proc. ACM
  Int'l Conf. on Multimedia (MM)}, Seoul, Republic of Korea, Oct. 2018.

\bibitem{sun2019two}
L.~Sun, F.~Duanmu, Y.~Liu, Y.~Wang, Y.~Ye, H.~Shi, and D.~Dai, ``A two-tier
  system for on-demand streaming of 360 degree video over dynamic networks,''
  \emph{IEEE Journal on Emerging and Selected Topics in Circuits and Systems},
  vol.~9, no.~1, pp. 43--57, Mar. 2019.

\bibitem{yuan2019spatial}
H.~Yuan, S.~Zhao, J.~Hou, X.~Wei, and S.~Kwong, ``Spatial and temporal
  consistency-aware dynamic adaptive streaming for 360-degree videos,''
  \emph{IEEE J. Sel. Topics Signal Process.}, 2019 (Early Access).

\bibitem{LeFeuvre2016-Tiled}
J.~Le~Feuvre and C.~Concolato, ``Tiled-based adaptive streaming using
  {MPEG-DASH},'' in \emph{Proc. ACM Int'l Conf. on Multimedia Systems (MMSys)},
  Klagenfurt, Austria, May 2016.

\bibitem{shi2019mobile}
S.~Shi, V.~Gupta, M.~Hwang, and R.~Jana, ``Mobile {VR} on edge cloud: {A}
  latency-driven design,'' in \emph{Proc. ACM Int'l Conf. on Multimedia Systems
  (MMSys)}, Amherst, MA, Jun. 2019.

\bibitem{Nguyen2019-an}
D.~V. {Nguyen}, H.~T.~T. {Tran}, A.~T. {Pham}, and T.~C. {Thang}, ``An optimal
  tile-based approach for viewport-adaptive 360-degree video streaming,''
  \emph{IEEE Trans. Emerg. Sel. Topics Circuits Syst.}, vol.~9, no.~1, pp.
  29--42, Mar. 2019.

\bibitem{Zhang2018}
Y.~Zhang, P.~Zhao, K.~Bian, Y.~Liu, L.~Song, and X.~Li, ``{DRL360}: 360-degree
  video streaming with deep reinforcement learning,'' in \emph{Proc. IEEE
  INFOCOM}, Paris, France, Apr. 2019.

\bibitem{pang2019towards}
H.~Pang, C.~Zhang, F.~Wang, J.~Liu, and L.~Sun, ``Towards low latency
  multi-viewpoint 360${}^{\circ}$ interactive video: A multimodal deep
  reinforcement learning approach,'' in \emph{Proc. IEEE INFOCOM}, Paris,
  France, Apr. 2019.

\bibitem{hazan2016introduction}
E.~Hazan, ``Introduction to online convex optimization,'' \emph{Foundations and
  Trends in Optimization}, vol.~2, no. 3-4, pp. 157--325, Aug. 2016.

\bibitem{shalev2012online}
S.~Shalev-Shwartz, ``Online learning and online convex optimization,''
  \emph{Foundations and Trends in Machine Learning}, vol.~4, no.~2, pp.
  107--194, Feb. 2011.

\bibitem{Mokhtari2016-online}
A.~{Mokhtari}, S.~{Shahrampour}, A.~{Jadbabaie}, and A.~{Ribeiro}, ``Online
  optimization in dynamic environments: Improved regret rates for strongly
  convex problems,'' in \emph{Proc. IEEE CDC}, Las Vegas, NV, Dec. 2016.

\bibitem{Chen2017-optimization}
T.~Chen, Q.~Ling, and G.~B. Giannakis, ``An online convex optimization approach
  to proactive network resource allocation,'' \emph{IEEE Trans. Signal
  Process.}, vol.~65, no.~24, pp. 6350--6364, Dec. 2017.

\bibitem{paternain2017online}
S.~Paternain and A.~Ribeiro, ``Online learning of feasible strategies in
  unknown environments,'' \emph{IEEE Trans. Autom. Control}, vol.~62, no.~6,
  pp. 2807--2822, Nov. 2017.

\bibitem{Hall2015online}
E.~C. {Hall} and R.~M. {Willett}, ``Online convex optimization in dynamic
  environments,'' \emph{IEEE J. Sel. Topics Signal Process.}, vol.~9, no.~4,
  pp. 647--662, Jun. 2015.

\bibitem{jadbabaie2015online}
A.~Jadbabaie, A.~Rakhlin, S.~Shahrampour, and K.~Sridharan, ``Online
  optimization: Competing with dynamic comparators,'' in \emph{Proc. Int'l
  Conf. on Artificial Intelligence and Statistics (AISTATS)}, San Diego, CA,
  May 2015.

\bibitem{van2016http}
J.~Van Der~Hooft, S.~Petrangeli, T.~Wauters, R.~Huysegems, P.~R. Alface,
  T.~Bostoen, and F.~De~Turck, ``{HTTP/2}-based adaptive streaming of {HEVC}
  video over {4G/LTE} networks,'' \emph{IEEE Commun. Lett.}, vol.~20, no.~11,
  pp. 2177--2180, Nov. 2016.

\bibitem{Knorr2018data}
S.~Knorr, C.~Ozcinar, C.~O. Fearghail, and A.~Smolic, ``Director's cut: A
  combined dataset for visual attention analysis in cinematic {VR} content,''
  in \emph{Proc. ACM SIGGRAPH}, London, United Kingdom, Dec. 2018.

\bibitem{Huang2014BAR}
T.-Y. Huang, R.~Johari, N.~McKeown, M.~Trunnell, and M.~Watson, ``A
  buffer-based approach to rate adaptation: Evidence from a large video
  streaming service,'' in \emph{Proc. ACM SIGCOMM}, Chicago, IL, Aug. 2014.

\bibitem{speedtest-canada}
Speedtest, ``{Speedtest market report: Canada average mobile download speed
  based on Q2-Q3 data in 2019},''
  \url{https://www.speedtest.net/reports/canada/}, {Accessed on} Feb. 10, 2020.

\bibitem{bubeck2014theory}
S.~Bubeck, ``Theory of convex optimization for machine learning,'' \emph{arXiv
  preprint arXiv:1405.4980}, 2014.

\bibitem{youtube-dash}
{YouTube Help}, ``Recommended upload encoding settings,''
  \url{https://support.google.com/youtube/answer/1722171?hl=en}, {Accessed on}
  Feb. 10, 2020.

\bibitem{de2012quantifying}
T.~De~Pessemier, K.~De~Moor, W.~Joseph, L.~De~Marez, and L.~Martens,
  ``Quantifying the influence of rebuffering interruptions on the user's
  quality of experience during mobile video watching,'' \emph{IEEE Trans.
  Broadcast.}, vol.~59, no.~1, pp. 47--61, Mar. 2012.

\end{thebibliography}

\end{document}